\documentclass[journal]{IEEEtran}

\usepackage{tikz}

\usepackage{xcolor}

\usepackage{amsmath}
\usepackage{amsthm}
\usepackage{amssymb}
\usepackage{pifont}
\usepackage{multirow}
\usepackage{graphicx} 
\graphicspath{{./figures/}}

\ifCLASSOPTIONcompsoc
\usepackage[caption=false,font=footnotesize,labelfon
t=sf,textfont=sf]{subfig}
\else
\usepackage[caption=false,font=footnotesize]{subfi
g}
\fi

\usepackage{url}
\usepackage{booktabs} 
\usepackage{array} 
\usepackage{verbatim} 
\usepackage{enumitem}
\usepackage{cite}
\usepackage{bbm}

\usepackage[ruled,linesnumbered]{algorithm2e}

\usepackage{algorithmicx}
\usepackage{algpseudocode}

\usepackage[export]{adjustbox}

\usepackage[para]{threeparttable}

\usepackage{xcolor}
\definecolor{myblue1}{rgb}{0, 0.2784, 0.6705}
\newcommand{\change}[1]{\textcolor{myblue1}{#1}}

\usepackage{hyperref}
\begin{document}

\title{An Approximation Algorithm for Joint Caching and Recommendations in Cache Networks}


\author{Dimitra Tsigkari and Thrasyvoulos Spyropoulos
\thanks{The current version has been accepted for publication in {\sc IEEE Transactions on Network and Service Management}, \href{https://ieeexplore.ieee.org/document/9711560}{DOI: 10.1109/TNSM.2022.3150961}, date of acceptance: February 6, 2022.}
\thanks{The authors are with Eurecom, Biot, France. Email: \{dimitra.tsigkari, thrasyvoulos.spyropoulos\}@eurecom.fr.}
\thanks{Part of this work has appeared in the proceedings of the IEEE International Symposium on a World of Wireless, Mobile and Multimedia Networks~(WoWMoM) 2020.~\cite{tsigkari2020wowmom}}
\thanks{ This work has been supported by the French National Research Agency under the ``5C-for-5G'' JCJC project with ref. number ANR-17-CE25-0001.}}

\maketitle

\begin{abstract}
Streaming platforms, like Netflix and YouTube, strive to offer high streaming quality (SQ), in terms of bitrate, delays, etc., to their users. Meanwhile, a significant share of content consumption of these platforms is heavily influenced by recommendations. In this setting, the user's overall experience is a product of both the user’s interest in a recommended content, \emph{i.e.}, the recommendation quality (RQ), and the SQ of this content. However, network decisions (like caching) that affect the SQ are usually made without considering the recommender's actions. Likewise, recommendations are chosen independently of the potential delivery quality. In this paper, we define a metric of streaming experience (MoSE) that captures the fundamental tradeoff between the SQ and RQ. We aim to jointly optimize caching and recommendations in a generic network of caches, with the objective of maximizing this metric. This is in line with the recent trend for content providers to simultaneously act as Content Delivery Network owners, implying that the same entity may handle both caching and recommendation decisions. We formulate this joint optimization problem and prove that it can be approximated up to a constant factor. To the best of our knowledge, this is the first polynomial algorithm to achieve a constant approximation ratio for the joint problem. Moreover, our numerical experiments show important performance gains of our algorithm over baseline schemes and existing algorithms in the literature.
\end{abstract}

\begin{IEEEkeywords}
caching, recommendation systems, wireless networks, multimedia streaming services
\end{IEEEkeywords}

\section{Introduction}
\label{sec:intro}
\subsection{Motivation}
On platforms of streaming services such as YouTube, Netflix, and Spotify, state-of-the-art recommendation systems are employed in order to help the users to navigate through their catalogues. Traditionally, the goal of the recommender is to offer personalized recommendations based on the user's interests, {\it i.e.}, present contents of an ever-evolving catalogue that are relevant to her tastes. Therefore, the user's engagement with the service is closely related to the  recommendation quality~(RQ)~\cite{gomez2016netflix}.
Moreover, these recommendations are responsible for a large share of the generated user requests: on Netflix, $80\%$ of requests come from the recommendations that appeared to the user \cite{gomez2016netflix}.

At the same time, the streaming quality~(SQ) of the delivered content plays a significant role in the overall user's experience on the service. This quality can be characterized, for example, by metrics of quality of service~(QoS) such as bitrate, initial delays, etc.,  or/and by metrics of quality of experience~(QoE), \emph{i.e.,} QoS as perceived by the user. 
Furthermore, the SQ is closely related to the user engagement to the service. It has been shown that, on platforms of video streaming services, low bitrate can lead to an increase in the abandonment rate~\cite{Nam2016_catvideos}. When a user abandons a viewing session, or the service all together, there is an associated revenue loss for the Content Provider~(CP). In fact, Akamai estimates this loss for platforms where ads are displayed (like YouTube) in~\cite{akamai_whitepaper_OTT}. Caching mechanisms within the Content Delivery Networks~(CDN) that store  contents in caches close to the user can ensure a  better delivery  in terms of SQ~\cite{Ott:caching} while  alleviating the backhaul link  traffic. In future architectures, envisioned for wireless networks beyond 5G, it has been widely argued that there will exist a very large number of  small caches placed at the edge of the network and the caching problem in this ``small cache'' setting has been the focus of entire lines of research~(\emph{e.g.},~\cite{femto_JOURNAL2013}  or see~\cite{JSAC:caching-survey} for a survey on this subject). Therefore, making the right caching decisions is crucial in today's and, more importantly, in future architectures. This is exactly where the influence of the recommendations on users' requests could come into the picture. For this reason, some recent works propose caching and recommendation policies that take into account their interplay \cite{cache-centric-video-recommendation , giannakas2018show-me-cache, sermpezis2018sch-jsac, chatzieleftheriou2019joint-journal, qi2018optimizing, liu2019deep, Chen:Joint-Globecom18}. In fact,  it has been argued that such an approach could be beneficial not only for the network's performance but also for the user~\cite{kastanakis2020network}.


At first glance, content caching and recommendation systems seem to be independent, since they are usually handled by two different entities: the CP  and a 3rd party CDN~(like Akamai). However, major CPs like Netflix and Google started partnering with Internet Service Providers (ISPs) to implement their own CDN solutions inside the network: Netflix Open Connect and Google Global Cache. In particular, Netflix partners with ISPs and provides them with Open Connect caches that are  implemented within the ISP network in order to localize the traffic as close to the user as possible~\cite{netflix_overview}. Netflix is responsible for filling the caches during off-peak hours~\cite{netflix_fill_patterns} and, of course, for choosing user recommendations~\cite{gomez2016netflix}. Therefore, caching and recommendation decisions  can already be handled jointly in today's architectures. The trend towards end-to-end network slicing in future wireless networks further supports such an approach. In this setting, CPs like the aforementioned ones  will own their own virtual network (slice) including communication, storage, and CPU resources at the base stations of the Radio Access Network~(RAN). Moreover, handling caching and recommendations together seems appropriate in future architectures like the ``femtocaching'' scenario~\cite{femto_JOURNAL2013}. In this setting, each cache will cover fewer users (than today's CDN caches) and receive a limited amount of requests. This would render predicting content popularities and thus caching decisions particularly challenging~\cite{Paschos-infocom2016}.

Motivated by the above, in this paper, we define a metric of streaming experience~(MoSE) that captures the user's interest in a recommended content and in the quality that this will be delivered. The MoSE is modeled as a function of caching and recommendation decisions. Next, we address the problem of maximizing users' MoSE in a setup where a CP jointly  controls these decisions. Furthermore, a strong motivation behind our work is the fact that Netflix  steers user recommendations towards popular (and possibly cached) contents. This is done, for example, through the ``Trending Now” section of every user’s main page. In particular, ``Trending Now" recommendations include short-term popular contents “combined with the right dose of personalization”~\cite{gomez2016netflix}. Meanwhile, at Netflix, caching decisions are based on (predicted) content popularities~\cite{netflix_and_fill}. 
Besides, we will show later that caching without taking into account the user recommendations and biasing recommendations later (to favor cached contents) is suboptimal.  This further corroborates our approach of jointly optimizing caching and recommendations.   Moreover, a recent development that supports such an approach is that Netflix mobile app introduced the feature ``Downloads for You''~\cite{netflix_downloads_for_you}. When a user enables this feature, she allows Netflix to choose and proactively prefetch recommended contents on the user's device that the user can then watch while she is offline or over unstable Internet connections.

Despite the arguments given above supporting a joint approach, many of the related works on caching and recommendations still focus on one side of the problem, {\it e.g.}, network-friendly recommendations~\cite{cache-centric-video-recommendation},~\cite{giannakas2018show-me-cache}, or recommendation-aware caching policies~\cite{sermpezis2018sch-jsac}. Some works that do try to modify both  caching and recommendations are usually based on heuristics~\cite{chatzieleftheriou2019joint-journal, qi2018optimizing, liu2019deep,  Chen:Joint-Globecom18}.  In addition,  most of the aforementioned works study the problem from the CP's point of view, {\it i.e.}, with the goal of maximizing cache hit rate, without taking into account the user's experience.

\subsection{Our approach and contributions}
In this paper,  we formulate and analytically study the problem of jointly optimizing both variables : (i) what content to store at each cache, and (ii) what content to recommend to each user, based on their location in the caching network and their predicted preferences.
In this direction, our main contributions are the following:
\begin{itemize}[leftmargin=*]
\item  We introduce the  metric of streaming experience~(MoSE) for a recommendation-driven content application that is expressed as a balanced sum of SQ~(as a function of the caching variables) and RQ~(as a function of the recommendation variables).  Based on this model, we formulate the problem of optimally choosing both sets of variables towards maximizing users' MoSE. 
    
\item The joint caching and recommendation problem has been shown before to be NP-hard and, thus, it cannot be solved optimally  in polynomial time (unless $P=NP$). 
However, we provide a polynomial-time algorithm  that approximates the optimum objective function value within a constant factor.
    
\item  Through numerical evaluations, we show a near-optimal performance and significant gains of our proposed policy over a variety of baseline schemes and existing heuristics for the joint problem.  We also implemented distributed versions of our policy and we show that important speedups can be achieved. Our evaluations were conducted on both real and synthetic datasets, and using realistic values for the problem parameters.

\end{itemize}

This paper extends our earlier work~\cite{tsigkari2020wowmom} by providing the proofs of the theoretical results and a comprehensive evaluation of our proposed algorithm in a variety of scenarios and for different input parameters. Finally, the current work discusses in depth the applications of the presented problem and the possible extensions.

\section{Problem Setup}
\label{sec:problem_setup}
\subsection{Caching Network}
We consider a set of $C$ caches with capacity $\mathcal{C}_j$, $j=1, \ldots, C$ and a content catalogue  $\mathcal{K}$.  Moreover, $  \mathcal{C}_j \ll |\mathcal{K}|$, $j=1, \ldots, C$, as is an important restriction in most caching setups\footnote{In fact, according to estimations of the size of Netflix catalogue \cite{Paschos-misconceptions} and the size of Open Connect appliances \cite{netflix_appliances},  the cache capacity of the appliances varies from $0.1\%$ to $2.3\%$ of the entire catalogue. \label{footnote-cache-capacity}} and especially in future wireless networks. We will consider both equal and variable-sized contents. In the second case, we denote by $\sigma_i$ the size of content $i$, where $i=1,\ldots, |\mathcal{K}|$.


\theoremstyle{definition} \newtheorem{def_caching}{Definition} 
\begin{def_caching}[Caching variables]
We let $x_{ij}$ be the binary variable, where $x_{ij}=1$ when the content $i$ is cached in cache $j$, and $x_{ij}=0$ otherwise.  We denote the corresponding matrix by $X=\{x_{ij}\}_{i,j}$.
\end{def_caching}


We consider a set $\mathcal{U}$ of users, each of which has access to a subset of caches. We denote this set by $\mathcal{C}(u)$ for user $u\in \mathcal{U}$. A request for content $i$ by  user $u$ is served by one of the caches belonging to $\mathcal{C}(u)$ where the requested content is stored, \textit{i.e.}, by one of the caches of the set
$ \{j: j\in \mathcal{C}(u) \text{ and } x_{ij}=1 \}$.
The access to a cache could be over multiple links (as in hierarchical caching or ICNs) or direct ({\it e.g.}, wireless connectivity to a nearby small cell~\cite{femto_JOURNAL2013}). For the purposes of our analysis, such networks can be represented as a generic bipartite graph between users and (associated) caches, as shown in Fig.~\ref{figure:illustration-example}. Specifically,  every edge of this graph has a weight $s_{uj}$, which denotes the cache-specific streaming quality\footnote{This assumption is realistic. At Netflix, for example, information on routes, network proximity to the users, etc. are gathered by the Open Connect caches and are sent regularly to the Netflix cloud~\cite{netflix_overview}.} that can be supported between user $u$ and cache $j$.  This quality  can be related to estimations on the QoS or QoE, \emph{i.e.,} delays, rebufferings, rate switching, etc. Moreover, it 
may differ from cache to cache, or may depend on channel quality, number of hops, scheduling policy, congestion level, etc.
Finally,  there is a large cache $C_0$ that fits all the  contents, {\it i.e.}, $x_{i0}=1$ for all $i\in \mathcal{K}$, and is accessible by all users, {\it i.e.}, $C_0 \in \mathcal{C}(u)$, for all $u\in \mathcal{U}$. This could be a large cache deep(er) in the network. For this reason and w.l.o.g., we let $s_{u0} < s_{uj}$, for all $j$ and $u$, as is commonly assumed ({\it e.g.}, in \cite{femto_JOURNAL2013}, \cite{poularakis2014}).

In our model, the caches are filled or updated during off-peak hours and, therefore, the cache allocation is static for the time period between two cache updates~(\emph{e.g.,} a day or a time window of a few hours). In this context, if a requested content is not cached in any of the caches, the content is served by the  large cache $C_0$. Although CDN caches traditionally have been operated with the use of dynamic caching policies~(where the cache is typically updated upon a cache miss), it has been argued that such policies perform poorly in the scenario of small caches~(under a non-stationary request model)~\cite{Paschos-infocom2016}. Hence, in such setups, similar models like ours are much more common in related works, {\it e.g.}, \cite{femto_JOURNAL2013}, \cite{poularakis2014}. More importantly, there is a trend towards static caching model even in today's architectures: Netflix, for example, is updating the Open Connect caches every night during off-peak hours~\cite{netflix_fill_patterns}. This further supports our model. Therefore, in what follows, all the problem parameters are  considered to be known for the time period between two cache updates. 
The presented caching setup is generic and could capture a variety of caching networks, such as femto-caching setups~\cite{femto_JOURNAL2013}, hierarchical CDN networks~\cite{borst2010}, etc. 

\subsection{Recommendations}

A list of $N_u$ recommended contents appears to the user $u\in \mathcal{U}$. This number may vary from user to user depending on the device used, as is the case in Netflix~\cite{gomez2016netflix}, for example.  The  recommendations are personalized and  might depend on  various factors such as user ratings ({\it e.g.}, via collaborative filtering), past user behavior, viewing times, etc.~\cite{adomavicius_rec_ranking}. State-of-the-art recommenders usually  assign a utility (or ``score'' or ``rank'') to each content for each user $u$~\cite{adomavicius_rec_ranking}, \cite{amatriain_building}. We denote by $r_{ui}\in [0,1]$ these utilities. Typically, the CP would select the $N_u$ items with the highest $r_{ui}$ to feature the recommendations list of user $u$. In this work, the recommendation decisions~(\emph{i.e.,} deciding which contents will appear to the user's recommendations list) are made not only based on the utilities $r_{ui}$ but also on the caching decisions. Therefore, our model uses the utilities $r_{ui}$~(that are derived from a typical recommender) as input for our problem.


\theoremstyle{definition} \newtheorem{def_recomm}[def_caching]{Definition}
\begin{def_recomm} [Recommendation variables]
 We let $y_{ui}\in \{0,1\}$ denote the binary variable for content $i$ being recommended  to user $u$ $(y_{ui}=1)$ or not $(y_{ui}=0)$. We denote by $Y$ the  matrix of $y_{ui}$. Then, the equations $\sum_{i\in \mathcal{K}} y_{ui} = N_u$, for all $u\in \mathcal{U}$, capture the fact that $N_u$ contents are recommended.
\end{def_recomm}

 \emph{Motivated by the discussion in Section~\ref{sec:intro}, we assume that both caching and recommendation decisions are made by the same entity ({\it e.g.}, Netflix).}

 \subsection{User model}  \label{subsec:setup_usermodel}
The user makes content requests, affected by the aforementioned recommendations, according to the following  model:

\begin{itemize}
    \item with probability $\alpha_u$ the user requests a recommended content. Each of the $N_u$ recommended items will be chosen with equal probability by the user;
    \item with probability $(1-\alpha_u)$ the user  ignores the recommendations and request a content $i$ of the catalogue with  probability $p_{ui}$.
    
\end{itemize}
Essentially, $\alpha_u$ captures the percentage of time a user $u$ tends to follow the recommendations. For example, it is estimated, on average, that $\alpha_u = 0.8$  on Netflix~\cite{gomez2016netflix}, but it can of course differ among users. 
 Assuming prior knowledge of  the user's disposition to follow the recommendations is common in related works ({\it e.g.}, \cite{giannakas2018show-me-cache}, \cite{chatzieleftheriou2019joint-journal}) and also in other  works on recommendation systems ({\it e.g.}, \cite{bressan_rec}). 
In practice, $\alpha_u$ might change over longer time intervals both  because of intrinsic changes to user behavior or due to decreasing/increasing trust in the recommender. Such changes could be addressed by dynamic or stochastic models, that are out of the scope of this work. In this work, we assume that our optimization happens at a smaller time scale, for which  the   parameter $\alpha_u$ is roughly constant~(but it can be recalibrated at longer intervals).

Furthermore, the assumption that each recommended content will be clicked with equal probability $1/N_u$ is also common in related works, and might hold in scenarios where the recommended items are ``unknown'' to the user, and hence she cannot evaluate their utility, before requesting them. 

As for the $p_{ui}$, they  capture the probability of user $u$ requesting the content $i$ outside of recommendations ({\it e.g.}, through the search bar). This could be an arbitrary distribution over the catalogue ({\it e.g.}, with probability mass only on content the user already ``knows''). Alternatively, given the utilities $r_{ui}$, a reasonable choice could also be the normalized values:

\begin{equation} \label{popularities_r} 
   p_{ui}=r_{ui}/ \sum_{k\in \mathcal{K}} r_{uk}. 
\end{equation}


 \subsection{Example} \label{subsec:setup_example}

   \begin{figure}[th]
\centering  
\includegraphics[width=8.85cm, trim={1.53cm 12cm 1.9cm 2.5cm},clip]{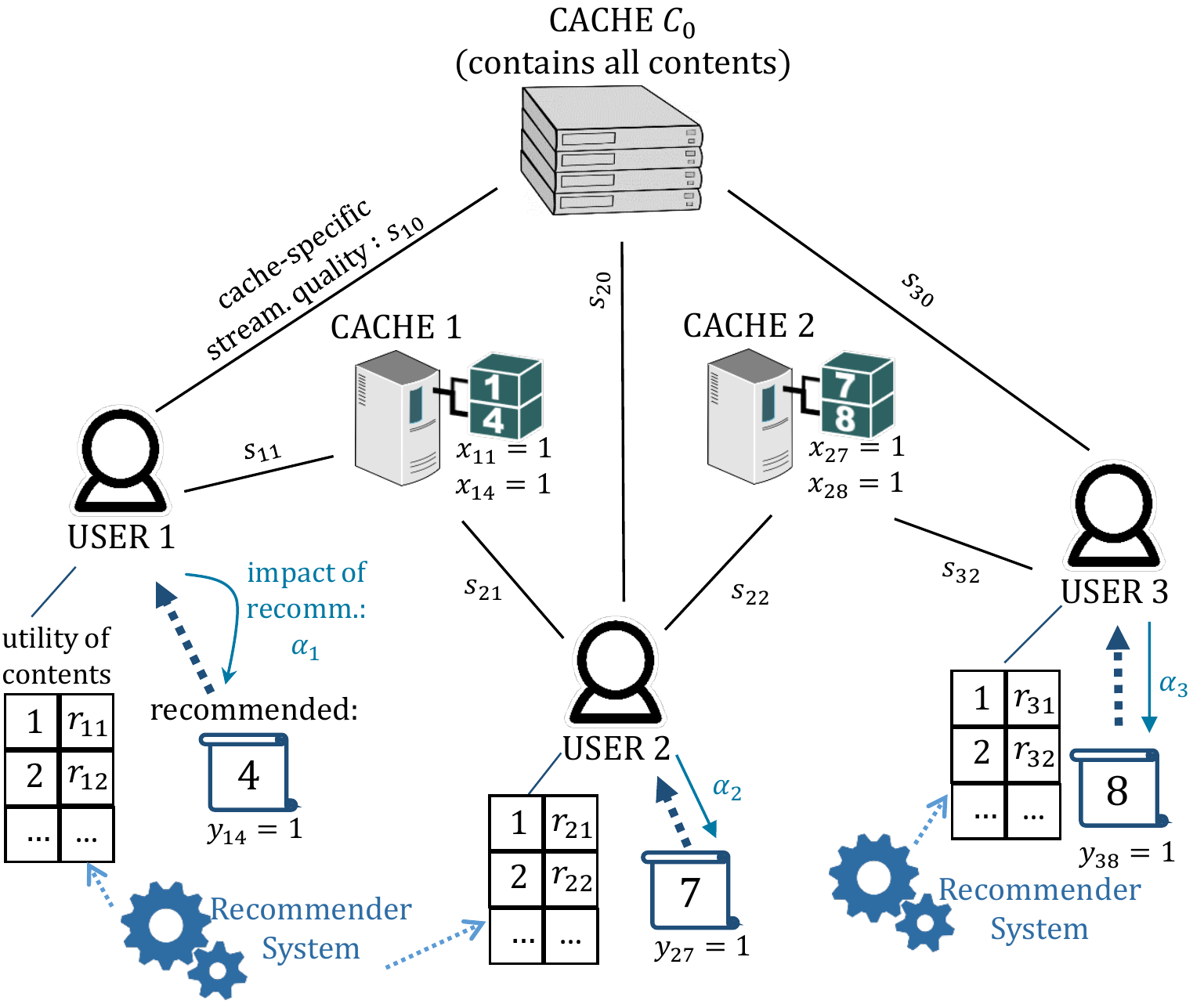} 
\caption{Illustration of the variables and  parameters considered for the joint caching and recommendations problem in a network of caches. Detailed description of this example can be found in Section~\ref{subsec:setup_example}.}  \label{figure:illustration-example}
\end{figure}

To better elucidate our model thus far, we present a small-scale example and Fig.~\ref{figure:illustration-example} that illustrates the variables and the parameters defined above. We consider a network of $C=2$ caches of capacity $2$, and a large cache $C_0$ containing the entire catalogue that consists of $|\mathcal{K}|=9$ equal-sized contents. As shown, cache $1$ contains contents $1$ and $4$ ({\it i.e.}, $x_{11}, x_{14}=1$), while $x_{1j}=0$ for any other $j$.  There are $3$ users present in the network.  An edge between a user $u$ and a cache $j$ means that user $u$ can fetch a content from cache $j$.   For example, for user $1$, we have that $\mathcal{C}(1) = \{0,1\}$. Note that such an edge might actually correspond to a path of multiple physical links. The corresponding cache-specific stream. quality  is indicated as an edge weight. In this example, a single  recommendation ($N_u=1$) appears to every user (illustrated by a dashed-line arrow). For example, the content $4$ is recommended to user $1$ ({\it i.e.}, $y_{14}=1$). If user~$1$ requests it, then it can be streamed from cache $1$ at streaming quality $s_{11}$. However, if user~$1$ requests, say, the content $2$, this will be fetched  from cache $C_0$ at a (lower) quality $s_{10}$. Lastly, arrows from users to recommendations display the probabilities $\alpha_u$.

 \begin{table}[!h] 
\caption{Notation summary} \label{notation_summary}
\centering 

\begin{tabular}{|c|l|}
\hline
\textbf{Notation} & \textbf{Description} \\
\hline
$\mathcal{K}$ & catalogue of contents\\
\hline
$\mathcal{U}$ & set of users in the network\\
\hline
$C_0$ & large cache containing the entire catalogue\\
\hline
$C$ & number of caches in the network ($C_0$ is excluded)\\
\hline

$\mathcal{C}_j$ & capacity of cache $j$, $j=0,\cdots,C$\\
\hline
$\mathcal{C}(u)$ & set of caches that user $u$ communicates with \\
\hline
$r_{ui}$ & utility of content $i$ for user $u$\\
\hline
$s_{uj}$  & cache-specific streaming quality between user $u$ and \\
($s_{u(j)}$)  &  cache $j$, $s_{u(j)}$ for ordered qualities (see Def.~\ref{def:order_statistics})\\
\hline
$\sigma_i$ & size of content $i$\\
\hline
$N_u$ & number of recommended contents for  user $u$\\
\hline

$\alpha_u$ & probability that user $u$ follows the recommendations \\
\hline

\multirow{2}{*}{$p_{ui}$} & probability that user $u$ requests content $i$   while not\\
& following the recommendations \\
\hline
\multirow{2}{*}{$\varphi$} & function of $r_{ui}$ that captures the impact of $r_{ui}$\\
& in the perceived recommendations quality\\
\hline
\multirow{2}{*}{$x_{ij}$} & caching variable, $x_{ij}=1$ when content $i$ is cached\\
&  in cache $j$, and $x_{ij}=0$ otherwise\\
\hline
\multirow{2}{*}{$y_{ui}$} & recommendation variable, $y_{ui}=1$ when content $i$\\ & is recommended to user $u$, and $y_{ui}=0$ otherwise\\
\hline

\end{tabular}
\end{table} 
\subsection{Metric of Streaming Experience~(MoSE)} \label{subsec:QoE}
In the context of media streaming platforms, the user's entertainment and contentment with the provided services are affected by the quality of the recommendations she receives, {\it i.e.}, if they are tailored to her tastes or not. On the other hand, it has been observed that low SQ ({\it e.g.}, low bitrates, rebufferings, etc.) greatly affects user experience and, most importantly (for CPs), retention/abandonement rates~\cite{Nam2016_catvideos}. In fact, experiments on video streaming services have shown that the user's overall experience depends on both the streaming quality and the user's interest in a content~\cite{li2016impact}. Moreover, some recent experimental evidence suggests that users might be willing to tradeoff (some) content relevance for (better) QoS\cite{kastanakis2020network}. In this direction, \emph{we define the metric of streaming experience as a twofold quantity}: one part relates to the recommendation quality;
the second part relates to  the streaming quality.

\theoremstyle{definition} \newtheorem{def_quality}[def_caching]{Definition}
\begin{def_quality}[Recommendations Quality - RQ]  \label{def_quality}
The recommendations quality, as perceived by user $u$, is equal to $\sum_{i \in \mathcal{K}} y_{ui} \varphi(r_{ui} )$, where $\varphi$ is any  non-decreasing function. 
\end{def_quality}

The function $\varphi$ represents the impact of a recommended content's utility $r_{ui}$ in the user's perceived RQ. It could be a linear function, or, more commonly, a concave function ({\it e.g.}, $\log(r_{ui})$) to capture diminishing returns beyond a minimum content utility. Moreover, we can demand a minimum quality $r_{\min}$ for any recommendation\footnote{This quantity can serve as an additional safeguard to support the earlier assumptions: {\it e.g.}, if any recommended content's utility is above $r_{min}$, then, $\alpha_u$, the user's trust in the recommendations,  will not be compromised.} if we define $\varphi$ as follows:
\begin{equation} \label{phi_rmin}
    \varphi(r_{ui}) =
    \begin{cases}
    \log(r_{ui}) & \text{if } r_{ui}\ge r_{\min},\\
    -\infty & \text{otherwise}.
    \end{cases}
\end{equation}

Regarding the SQ, this depends on which cache the requested content is streamed from. A content $i$ requested by user $u$ will be fetched by the ``best'' \emph{connected} cache that stores it,  as in~\cite{femto_JOURNAL2013}.

\theoremstyle{definition} \newtheorem{def:order_statistics}[def_caching]{Definition}
\begin{def:order_statistics} [Ordered cache-specific qualities] \label{def:order_statistics}
If $\mathcal{C}(u)$ is the set of caches that user $u$ has access to, we let $s_{u(1)}=\max\{s_{uj}, j \in \mathcal{C}(u)\}$ denote the maximum (cache-specific) quality for user $u$. Similarly, $s_{u(2)}$ denotes the second highest quality for $u$, and so forth\footnote{As the qualities $s_{uj}$ are sorted for every user, the notation $s_{u (k)_u}$ would be more appropriate. For simplicity, we drop the sub-index $u$.}.
\end{def:order_statistics}

 By definition, $s_{u|\mathcal{C}(u)|}=s_{u0}$, for every $u\in \mathcal{U}$, since we assumed that $s_{u0} < s_{uj}$, for all $j=1, \ldots, C$.  In the following lemma, the expected streaming quality~(SQ) is given as a function of the caching policy ($x_{ij}$), the values of $s_{uj}$, the recommendations ($y_{ui}$), and the users preferences ($r_{ui}$).

\theoremstyle{plain} \newtheorem{rate_formula}{Lemma} 
\begin{rate_formula}[(Expected) Streaming Quality- SQ] \label{rate_formula}
For a given cache allocation $X$, a content $i\in \mathcal{K}$ will be streamed to user $u$~(upon request) in the quality:
\begin{equation} \label{highest_rate}
  s_{u}(X,i):=  \sum_{j=1}^{|\mathcal{C}(u)|} \big[ s_{u(j)} x_{i(j)}  \prod_{l=1}^{j-1} (1-x_{i(l)}) \big],
\end{equation}
where $x_{i(j)}$ are similarly the caching variables assuming a $s_{uj}$-based ordering\footnote{Given the  $s_{uj}$-based ordering, $x_{i(j)}$ indicates if the  content $i$ is cached in the cache that offers the $j$-th highest quality for user $u$.}.
The expected streaming quality~(SQ) for a user $u$ is equal to: 
%
%

\begin{IEEEeqnarray}{lcl} \label{equation_avg_su}
\overline{s}_u = \alpha_u \sum_{i \in \mathcal{K}}  \dfrac{y_{ui}}{N_u} s_u(X,i) + (1-\alpha_u) \sum_{i \in \mathcal{K}}  p_{ui} s_u(X,i).\IEEEeqnarraynumspace
\end{IEEEeqnarray}
\end{rate_formula}

\begin{proof}
For a requested content $i\in \mathcal{K}$, $\prod_{l=1}^{j-1} (1-x_{i(l)})x_{i(j)}$ captures the fact that $i$ will be retrieved by the cache $(j)$ ({\it i.e.}, the cache with the $j$-th highest quality)  for lack of any other cache with higher quality  in $\mathcal{C}(u)$ where the content is cached~({\it i.e.}, $x_{i(l)}=0, l<j$). Then, this request will be served in the cache-specific quality $s_{u(j)}$.  Of course, if $i$ is not cached in any cache, it will be retrieved from  $C_0$ which is ranked last, resulting in low streaming quality.  Essentially,  $s_{u}(X,i)$ is the highest cache-specific quality associated to content $i$ for user $u$ among all the locations where $i$ is cached.
Finally, given that, upon request, the content $i$ will be streamed in the quality $s_u(X,i)$, and given the probabilities of such a request to happen~(through recommendations or not), the formula of the expected streaming quality, $\overline{s}_u$, easily follows.
\end{proof}

\theoremstyle{remark} \newtheorem{remark_psi}{Remark}
\begin{remark_psi} \label{remark:cachehit}
When estimating\footnote{Since our focus is proactive caching, we consider the estimated or average SQ. An interesting direction for future work would be to consider dynamic policies where the SQ is time-varying or even unknown.} the SQ, $s_{uj}$, the cache-specific streaming quality, can be chosen to be a function of QoS-related or QoE-related estimations. For example, $s_{uj}$ could be the estimated bitrate between user $u$ and cache $j$, or it could also include factors related to rebuffering probabilities, jitter, delays, etc., as is commonly considered in works related to streaming experience~\cite{batteram2010delivering}. Our framework is defined in such a way that is flexible enough to optimize whatever such  values(s) the CP deems appropriate or is able to estimate. 
Alternatively, $s_{uj}$ could be equal to:

\begin{equation} \label{psi_function_hit}
   s_{uj} =
    \begin{cases}
    1 & \text{if } j\in \mathcal{C}(u) \setminus C_0,\\
    0 & \text{otherwise}.
    \end{cases}
\end{equation}

\noindent In that case, $\sum_u \overline{s}_u$ estimates the expected cache hits for the (small) caches: upon a request, it counts  $1$ if the content is cached. In Table~\ref{tab:example_components}, we provide a variety of examples functions/values for $s_{uj}$ that exist already in the literature on multimedia streaming services.

\end{remark_psi}

\theoremstyle{definition} \newtheorem{definition:QoE}[def_caching]{Definition}
\begin{definition:QoE}[MoSE function] 
The metric of streaming experience for user $u\in \mathcal{U}$ as a function of the caching and recommendation variables is defined as 
$\overline{s}_u + \beta_u \sum_{i \in \mathcal{K}} y_{ui} \varphi(r_{ui} )$,
where $\overline{s}_u$ is given by \eqref{equation_avg_su} and $\beta_u > 0$ is a tuning parameter. 
Then the aggregate MoSE over all users is equal to:
\begin{equation} \label{def_QoE_function}
  f(X,Y):= \sum_{u \in \mathcal{U}}
 \big[  \overline{s}_u + \beta_u \sum_{i \in \mathcal{K}} y_{ui} \varphi(r_{ui} ) \big].
\end{equation}
\end{definition:QoE}

Modeling the streaming experience in this fashion implies a tradeoff between SQ and RQ, as discussed at the beginning of this section. The value of $\beta_u$ is the weight we attach to the RQ and quantifies  the importance of the RQ  compared to the SQ.  Moreover, the value of $\beta_u$  might differ from user to user: large $\beta_u$ means the user $u$ is more sensitive to the RQ, while small $\beta_u$ that she is more sensitive to the SQ. The choice of $\beta_u$ can depend, for example, on user behavior: we might want a small $\beta_u$~(\emph{i.e.,} priority to the SQ) for a user who often abandons the viewing session when the streaming quality is low. Similarly, one could imagine more complex models ({\it e.g.}, based on machine learning). However, it is beyond the scope of this paper to investigate such models or good choices for $\beta_u$, $\varphi$, and $s_{uj}$. Instead, our focus is to propose efficient algorithms for \emph{any} values and conforming functions.

We stress here that the MoSE was defined in such a generic way that can be adjusted according to the needs of the CP. Particularly, the CP can choose the quantities $s_{uj}$ and the function $\varphi$ based on available measurements in network, user behavior, etc. We provide a detailed example set of choices~(derived mostly from related work) for the different
MoSE components, as shown in Table~\ref{tab:example_components}.



\begin{table*}[h] 
  \centering

 \caption{ Example functions/values for MoSE components}  \label{tab:example_components}
  \begin{tabular}{|c|c|c|}
  \hline
  \textbf{Component} & \textbf{Example functions/values} & \textbf{Comments} \\
  \hline
  \begin{tabular}{c@{}c@{}}  $s_{uj}$ in SQ \\  (as in Lemma~\ref{rate_formula}) \end{tabular} & \begin{tabular}{p{0.32\textwidth}} \textbullet \quad QoS  metrics  \\ \hline \textbullet \quad estimated QoE as a function $\psi$ of QoS \\ \quad ($\psi$ can be linear, logarithmic, exponential etc.)  \\ \hline \textbullet \quad estimated multidimensional QoE \\ \hline \textbullet \quad cache hits (as defined in~\eqref{psi_function_hit})  \end{tabular} & \begin{tabular}{p{0.36\textwidth}}  \emph{e.g.,} related to bitrate, delays, packet loss, etc. \\ \hline  as in~\cite{fiedler2010generic}~(exponential QoS-QoE relation, IQX hypoth.)\\ or as in~\cite{logarithmic_QoE}~(logarithmic QoS-QoE relation)\\ \hline as in~\cite{skorin2012multi} or as in~\cite{batteram2010delivering} \\ \hline  requires no network-related measurements/estimations \end{tabular} \\
  \hline
  $\beta_u >0$ (weight factor) & \begin{tabular}{p{0.32\textwidth}}\textbullet \quad small $\beta_u$  \\ \hline \textbullet \quad large  $\beta_u$\end{tabular} &  \begin{tabular}{p{0.36\textwidth}} $\beta_u$ represents the sensitivity of user $u$ to the SQ and RQ, \\ small/large $\beta_u \rightarrow$ priority to the SQ/RQ\end{tabular}\\
  \hline
  \begin{tabular}{c@{}c@{}} function $\varphi$ in RQ \\ (as in  $\sum_{i \in \mathcal{K}} y_{ui} \varphi(r_{ui} )$) \end{tabular} & \begin{tabular}{p{0.32\textwidth}} \quad \\ \textbullet \quad logarithmic function\\ \hline \textbullet \quad  as defined in~\eqref{phi_rmin}\end{tabular} & \begin{tabular}{p{0.36\textwidth}} it captures the diminishing returns property \\ (as often observed in human perception~\cite{dehaene2003neural}) \\ \hline it captures minimum QR thresholds per user \end{tabular} \\
  \hline
  \end{tabular}
  
\end{table*}


\subsection{Joint recommendation and caching} \label{subsec:joint}
We ask the question: 
\emph{How can we make caching and recommendation decisions in order to maximize the MoSE?} 

To better understand the tradeoffs involved, we present a toy example depicted in Fig. \ref{toy_example}, and  two ``naive'' policies:
\theoremstyle{definition} \newtheorem*{Cpolicy}{Policy C, for ``Conservative''}
\begin{Cpolicy} 
This policy caches the $\mathcal{C}_j$ most popular contents (for the users connected to the cache $j$); it then recommends to each user $u$ the $N_u$ contents with the highest utility for this user, regardless of whether they are cached or not. This policy captures  today's status quo. \\
\noindent \textbf{Policy A, for ``Aggressive''.} This policy has the same caching policy as policy C, but it recommends only cached contents (the most relevant to the user  among them). It is closer to cache-friendly recommendation policies like the one in~\cite{sermpezis2018sch-jsac}.
\end{Cpolicy}

  \begin{figure}[htp]
\centering  
\includegraphics[width=8.6cm, trim={2.23cm 16.19cm 2.4cm 4.85cm},clip]{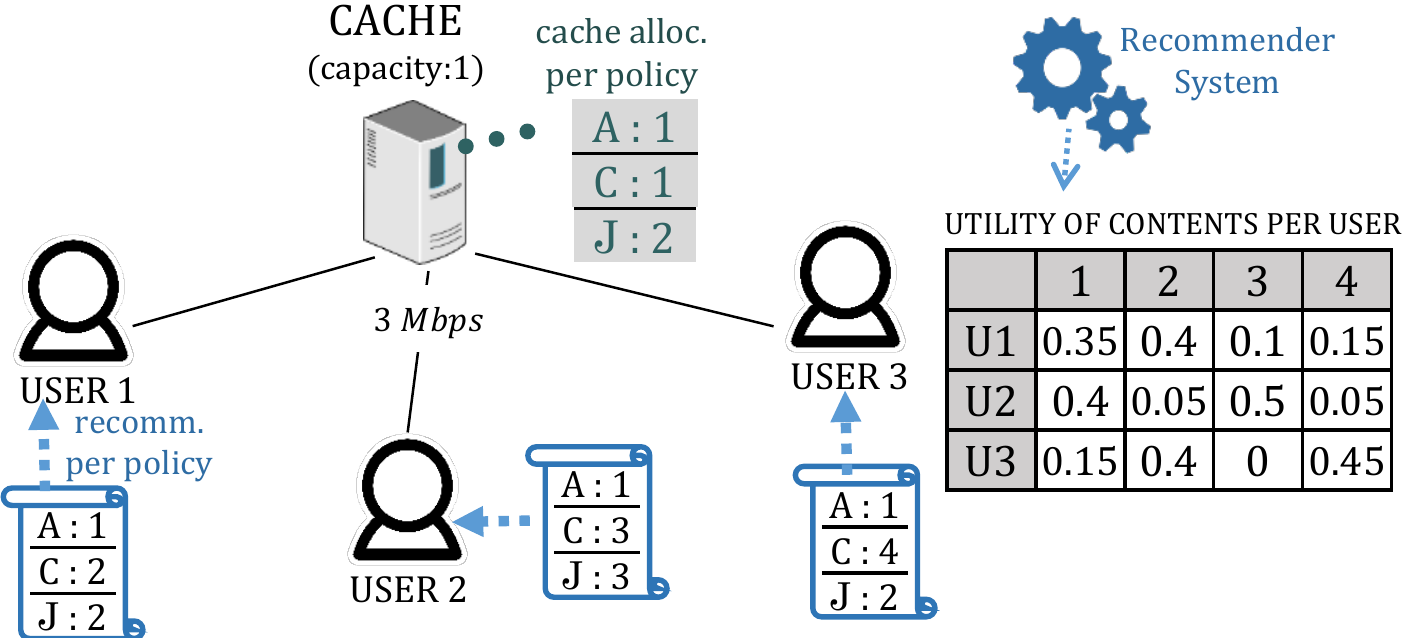}
\caption{Toy example of Sec.~\ref{subsec:joint}. On the left: illustration of the network together with the caching and recommendation decisions made by policies A, C, and J. On the right: the matrix of content utilities per user. } 
\label{toy_example}
\end{figure}

Note that both policies make the caching and recommendation decisions \emph{separately}. In this example, we will attempt to show the benefits of a policy that makes jointly these decisions. Referring to Fig. \ref{toy_example}, suppose we have a catalogue of $4$ equal-sized contents and $3$ users, all connected to the large cache $C_0$  (not shown in the figure, for simplicity) that contains all files and a smaller cache $C_1$ of capacity $1$. We measure the SQ through the estimated bitrate. All users can download a content from  $C_1$ or  $C_0$ with bitrate $3$ Mbps or $2$ Mbps respectively. 
We further assume that $N_u=1$ and  $\alpha_u =1$ for all users. We depict the utilities $r_{ui}$ on the right side.

Both policies would cache the  item with the highest aggregate utility, {\it i.e.}, content $1$. Policy A would then recommend this item to \emph{all} users. Policy C would instead recommend the item with highest utility per user, namely contents $2$, $3$ and $4$ respectively. It is easy to see that  policy C would lead to better RQ, while policy A would lead to better SQ.
However, we would like to know which policy is optimal with respect to maximizing the aggregate MoSE~(as expressed in~\eqref{def_QoE_function}).

A better option would be to cache content $2$, observing that this would then facilitate the recommendation decisions. More precisely, it allows one to recommend content $2$ to both users $1$ and $3$, achieving cache hits for them with maximum or close to maximum  RQ. Instead, for user $2$, the content $3$ is recommended~(with utility $r_{23} =0.5$), since content $2$ would seriously degrade the user's RQ~($r_{22}$ = 0.05 only). This policy which we refer to as ``J'' for Joint in Fig.~\ref{toy_example}, outperforms both A and C in this example in terms of MoSE~(for $\varphi$ being the logarithmic function and for most values of $\beta>0$).

In this example, it is easy to guess how to outperform the policies A and C (or even find the optimal one). However, this task becomes significantly harder for bigger scenarios
(when considering overlapping cache topologies, large content catalogues, multiple recommendations per user, etc.). 
To this end, in the next section, we formulate and analyse this problem, and propose an algorithm with approximation guarantees.

\section{Problem Formulation and analysis}
\label{sec:problem_formulation}
The optimization problem we are targeting is the following:

\theoremstyle{definition} \newtheorem*{QoE_problem_definition}{MoSE problem}
\begin{QoE_problem_definition} 
\begin{eqnarray} 
&\underset{X,Y}{\text{maximize  }}& f(X,Y)
\label{objective} \nonumber\\
&\text{subject to } & \sum_{i\in \mathcal{K}}  \sigma_i x_{ij} \leq \mathcal{C}_j \text{ for every }j= 1,\ldots, C;  \label{matroid_constraint} \\
&& \sum_{i \in \mathcal{K}} y_{ui} = N_u \text{ for every }u\in \mathcal{U} \label{N_recomm};\\
&& x_{ij}, y_{ui}\in \{0,1\}, \label{binary var}
\end{eqnarray}
\end{QoE_problem_definition}
\noindent where, according to~\eqref{highest_rate}, \eqref{equation_avg_su} and \eqref{def_QoE_function}, $f(X,Y)$ is equal to
\begin{equation*}
\sum_{u\in \mathcal{U}} \sum_{i\in \mathcal{K}} \big[\alpha_u \frac{y_{ui}}{N_u} s_u(X,i) + (1 - \alpha_u) \,p_{ui}\, s_u(X,i) + \beta_u y_{ui} \varphi(r_{ui})\big]
\end{equation*}
and $s_{u}(X,i):=  \sum_{j=1}^{|\mathcal{C}(u)|} \big[ s_{u(j)}\, x_{i(j)}  \prod_{l=1}^{j-1} (1-x_{i(l)}) \big]$ .
The constraints in \eqref{matroid_constraint} are the capacity constraints for every cache. In the case of equal-sized contents, \eqref{matroid_constraint} suggests that no more than $\mathcal{C}_j$  items  can fit in cache $j$, and the constraints in \eqref{N_recomm} suggest that each user receives $N_u$ recommendations. Finally, as expressed in \eqref{binary var}, the problem's variables, $x_{ij}$ and $y_{ui}$, are binary. 

\theoremstyle{plain} \newtheorem{NP-complete}[rate_formula]{Lemma} 
\begin{NP-complete} \label{NP-complete}
The MoSE problem is NP-hard.
\end{NP-complete}

\begin{proof}
An instance of the MoSE problem is the femto-caching problem in \cite{femto_JOURNAL2013} which is NP-hard.
\end{proof}

\subsection{Intuition on joint optimization}

As we saw in Lemma \ref{NP-complete}, even just the caching part ({\it i.e.}, maximizing in variable $X$) of the MoSE problem  is  hard to solve. For this simpler problem, the authors in~\cite{femto_JOURNAL2013} propose algorithms with approximation guarantees by exploiting submodularity properties of the objective. However, these algorithms do not account for the recommendation part of the MoSE problem (variable $Y$) and, therefore, the approximation guarantees do not extend to the joint problem. 

One could be tempted to extend the methodology in~\cite{femto_JOURNAL2013} by using both sets of variables $X$ and $Y$ as the ground set. However, the authors of~\cite{chatzieleftheriou2019joint-journal} prove that a subcase of the MoSE problem (when $\beta_u =0$)  is not submodular in $X$ and $Y$.

Furthermore, the authors of~\cite{sermpezis2018sch-jsac} consider  problem variants where the caching decision is ``recommendation-aware". They show that this problem is hard even for one cache, but manage  to retrieve submodularity properties and use the methodology of~\cite{femto_JOURNAL2013} to derive algorithms with approximation guaranties. However, their objective and problem setup do not contain recommendation variables~(among other things, the recommender's actions are fixed, and the caching policy simply knows what the recommender would do). It is thus significantly different than the MoSE problem. Finally, a brief qualitative comparison of these works is shown in Table \ref{tab:1}.

\begin{table}[!hbp] 
\centering 

\caption{State-of-the-art works on caching and/or recommendations}
  \begin{tabular}{|c|c|c|c|c|}
  \hline
    Related& \multicolumn{2}{c|}{Variables}  & How many &   Approx.  \\ 
    Works &Caching & Recomm. & caches?  &  guarantees \\
    \hline
\cite{femto_JOURNAL2013} & \ding{51} & \ding{55} & Network & \ding{51} \\
    \hline
\cite{sermpezis2018sch-jsac}& \ding{51} & \: \ding{55}$^{*}$ & Network & \ding{51} \\
    \hline
 \cite{chatzieleftheriou2019joint-journal} & \ding{51} &  \ding{51} &  Single cache & \ding{55} \\
    \hline
This work & \ding{51} &  \ding{51} &  Network & \ding{51} \\
    \hline
    
\multicolumn{5}{l}{ \small{${}^*$}\scriptsize{In \cite{sermpezis2018sch-jsac}, although the problem formulation does not contain any recom-}}  \\
    \multicolumn{5}{l}{\scriptsize{mendation variable, the caching variable is ``recommendation-aware".}}
  \end{tabular}

  \label{tab:1}
\end{table}


This discussion raises the question of whether the MoSE problem can be efficiently approximated and how. In the next section, we prove that this is indeed the case.  By first considering something akin to a primal decomposition~\cite{boyd_notes_decomposition} of the original problem (rather than handling variables $X$ and $Y$ at the same time as the ground set), we show that:
\begin{itemize}
    \item[(i)] for the problem on variables $Y$, {\it i.e.}, fixing $X$ (``inner" problem), the global maximizer can be found efficiently;
    \item[(ii)] the problem on variables $X$ (``outer" problem),  given the global maximizer of $Y$ (for any $X$), is in fact submodular.
\end{itemize}
 This property will  allow us to devise an algorithm for the joint problem that is polynomial in the problem size and, somewhat surprisingly, retains the approximation guarantees of the much simpler "caching-only" problems considered in~\cite{femto_JOURNAL2013} and~\cite{sermpezis2018sch-jsac}.

\subsection{Towards efficient algorithms} \label{effic_algor}

The key to our methodology is the following lemma. 
\theoremstyle{plain} \newtheorem{max_y_x}[rate_formula]{Lemma} 
\begin{max_y_x}\label{max_y_x}
The MoSE problem is equivalent to the problem:
\end{max_y_x}
\theoremstyle{definition} \newtheorem*{outer_problem}{Outer problem}
\begin{outer_problem} 
\begin{eqnarray} 
&\underset{X}{\mathrm{maximize }} & f^*(X):=f(X, \underset{Y}{\text{argmax}}f(X,Y)) \label{def_f*X}\\
&\text{subject to}& \eqref{matroid_constraint}, \eqref{N_recomm}, \text{ and }  \eqref{binary var}. \nonumber 
\end{eqnarray}
\end{outer_problem}

The equivalence of the two problems follows straightforwardly from the well known identity~\cite{boyd2004convex}: 
\begin{equation} \label{eq:decompos}
    \max_{X,Y} f(X,Y)= \max_X (\max_Y f(X,Y) ).
\end{equation}

\subsubsection{Inner problem and algorithm}

The first step would be to find a closed-form expression for $f^*$ for any cache allocation, {\it i.e.}, matrix $X$.
Hence, given  $X$, the problem of choosing the recommendation policy, {\it i.e.}, matrix $Y$,  is the problem of finding $f^*(X)$, as defined in \eqref{def_f*X}. We  formulate this problem:
\theoremstyle{definition} \newtheorem*{Recomm_problem}{Inner problem}
\begin{Recomm_problem} 
\begin{eqnarray*}
&\underset{Y}{\mathrm{maximize }} & f(X,Y)\\
&\text{subject to}&\eqref{N_recomm} \text{ and } y_{ui}\in \{0,1\}. \nonumber 
\end{eqnarray*}
\end{Recomm_problem}

\noindent The following lemma states that the inner problem can be decoupled into $|\mathcal{U}|$ problems.

\theoremstyle{plain} \newtheorem{lemma_max_y}[rate_formula]{Lemma} 
\begin{lemma_max_y} \label{lemma_max_y}
If $F_u^*(X):= \underset{Y}{\text{max  }}  \big( \overline{s}_u + \beta_u \sum_{i \in \mathcal{K}} y_{ui} \varphi(r_{ui} ) \big)$, for any $u$ and any placement $X$,  then $ f^*(X)= \sum_{u\in \mathcal{U}} F_u^*(X)$.
\end{lemma_max_y}

\begin{proof}
Given a cache placement $X$, it is easy to see that the recommendation decisions (variable $Y$) for a user do not interfere with the decisions for the other users. Moreover, the constraints in \eqref{N_recomm} are decoupled for every user.
\end{proof}

By \eqref{highest_rate} in Lemma \ref{rate_formula},  we can write $F_u^*(X)$ as follows.
\begin{eqnarray}
F_u^*(X)&=&
\underset{Y}{\text{max}} \bigg( \sum_{i\in \mathcal{K}} y_{ui} \big( \frac{\alpha_u}{N_u}  s_u(X,i) + \beta_u \varphi(r_{ui})\big) \bigg)\nonumber \\ &+& (1-\alpha_u)  \sum_{i\in \mathcal{K}} s_u(X, i) p_{ui}. \label{F_u_formula}
\end{eqnarray}

Next, we introduce the notion of V-value, which is the coefficient of  $y_{ui}$ in \eqref{F_u_formula}.
\theoremstyle{definition} \newtheorem{def:Rvalue}[def_caching]{Definition}
\begin{def:Rvalue} [V-value and ordered V-values] \label{def:Rvalue}
We define, as V-value of a content $i\in \mathcal{K}$ for user $u\in \mathcal{U}$ and for a given cache allocation $X$, the quantity
\begin{eqnarray} \label{def:eq:V}
V_{ui}(X):= \frac{\alpha_u}{N_u}  s_u(X,i) + \beta_u \varphi(r_{ui}), 
\end{eqnarray}
where $s_{u}(X,i)$ is defined in \eqref{highest_rate}. Similar to  Def. \ref{def:order_statistics},  we define the ordered $V_{ui}$ (sorted in decreasing order) as the ordered sequence $\{V_{u[k]}\}_{k\in \mathcal{K}}$\footnote{We do not use the same notation as in Def. \ref{def:order_statistics} because the ordering here is done with respect to the V-value and not the quality $s_{uj}$. In general, $V_{u(k)}(X) \neq V_{u[k]}(X)$, for all $u\in \mathcal{U}$ and $k=1,\ldots |K|$.}.
\end{def:Rvalue}

\theoremstyle{plain} \newtheorem{y_optimal}[rate_formula]{Lemma} 
\begin{y_optimal} \label{y_optimal}
For a given cache allocation $X$, we consider the matrix $Y'$ such that $y'_{u[k]}=1 $ for $k=1,\ldots, N_u$, and $y'_{u[k]}=0 $ otherwise, 
where $[k]$ is the content index associated to the $k$-th highest V-value for the user $u\in \mathcal{U}$. Then 
\begin{IEEEeqnarray}{lCl}
& F_u^*(X)&= \sum_{k=1}^{N_u} V_{u[k]}(X) +(1-\alpha_u)  \sum_{i\in \mathcal{K}} \big( s_u(X, i) p_{ui}\big), \label{F*formula} \IEEEeqnarraynumspace\\
 &\text{and }& f^*(X)=f(X, Y')= \sum_{u\in \mathcal{U}}  F_u^*(X). \label{extra_equality}
\end{IEEEeqnarray}
\end{y_optimal}

In words, the optimal solution for the inner problem is
to recommend to every user $u$ the $N_u$ contents with the highest V-value associated to the cache placement $X$. Note that this solution depends on the solution of the outer~(caching) problem and, hence, the inner problem needs to be solved as part of solving the outer problem, as shown in~\eqref{eq:decompos}.

\begin{proof}
It is straightforward to prove the result above through contradiction, {\it i.e.}, assuming some content $m$ with lower V-value than the $V_{u[N_u]}$ should have been included  instead.
\end{proof}

Based on Lemma~\ref{y_optimal}, the algorithm that finds the solution for the Inner Problem is summarized in Algorithm~1.

\begin{algorithm}[h]
\NoCaptionOfAlgo

\SetAlgoLined
\SetKwInOut{Input}{Input}
\SetKw{Return}{Return}
\SetAlgoLined

\Input{$\mathcal{U}$, $\mathcal{K}$, $N_u$, $X$, $\{\beta_u\}$, $\varphi$, $\{\alpha_u\}$, $\{r_{ui}\}$, $\{s_{uj}\}$}

Start with empty matrix $Y$ \\
 \For{every user $u\in \mathcal{U}$}{
  \For{ every content $i\in \mathcal{K}$}{
  Calculate $V_{ui}$\; 
  Store  $\{V_{u[k]}\}_{k=1}^{N_u}$~(in decreasing order). 
  \;
  }
  
  Set $y_{u[k]}=1$ for $k=1,\cdots, N_u$\;
 }
 \Return{$Y$}
 \caption{Algorithm 1: \textbf{Inner algorithm (subroutine)}}
\end{algorithm}

\subsubsection{Complexity of the inner algorithm}
The  internal for loop (lines $3-5$) consists of $|\mathcal{K}|$ calculations. 
Next, the complexity for the step of storing the $N_u$ highest  V-values is $O( \log N_u )$, however $N_u$ is considered to be a constant.
 Since these steps are repeated for every user, the total complexity of the inner algorithm is at most $O(|\mathcal{U}| \cdot |\mathcal{K}|)$.

\subsubsection{Outer problem and submodularity}
We proved that the optimal $Y$ can be found efficiently for the inner problem, given any cache allocation $X$. 
We want now to solve the outer problem~(defined in Lemma \ref{max_y_x}, \eqref{F*formula}, \eqref{extra_equality}), \emph{i.e.,} with respect to the caching variables $X$. While often caching problems can fit into the category of knapsack/general assignment problems, this is not the case for the outer problem. We note that the ``profit'' or gain of storing a content into a cache is not a constant and depends on the solution of the inner problem. Nevertheless, we will now prove some interesting properties of the outer problem that will lead us to an algorithm for the MoSE problem.

First, we extend $f^*$ as a set function. For any matrix $X$ we define  the corresponding placement $P_X$ of cached items:
\begin{equation*} 
P_X := \{ (i,j): x_{ij}=1, i\in \mathcal{K}, j=1, \ldots, C+1 \}.
\end{equation*}
Essentially, $P_X$ consists of the pairs (content, cache) of all the cached contents. Since, by definition, the large cache $C_0$ contains the entire catalogue ({\it i.e.}, $x_{i0}=1$, for all $i\in \mathcal{K}$), $X$ is  a $|\mathcal{K}| \times (C+1)$ matrix. 
In other words, $P_X$ belongs to the set     $\mathcal{P}:=P(\mathcal{K} \times \{1,\ldots,C+1\})$, 
 where $P(\mathcal{K} \times \{1,\ldots,C+1\})$ is the powerset of $\mathcal{K} \times \{1,\ldots,C+1\}$. Inversely, given a placement $P$, we can define the corresponding matrix $X_P$ such that $x_{ij}$ is equal to $1$, for every pair $(i,j)$ in $P$, and $0$ otherwise. Hence, from now on, $X$ and $P$ will be used interchangeably to denote the content allocation across the network of caches. 
We also define the subsets of a placement $P$  representing the storage of  the cache $m$:
 $P^{(m)}:= \{(i,m) \in P \}$.
We can thus extend  $F^*_u$, $f^*$, $s_u$ and $V_{ui}$ to the ground set $\mathcal{P}$. 

\theoremstyle{plain} \newtheorem{monotone}[rate_formula]{Lemma} 
\begin{monotone} \label{monotone}
The set function $F_u^*$ is monotone increasing for all $u\in \mathcal{U}$.
\end{monotone}
\begin{proof}
We consider two cache placements $P$ and $Q$ such that $P \subseteq Q \subseteq \mathcal{P}$ and we will prove that $F_u^*(P) \leq F_u^*(Q)$. Since  $P \subseteq Q$, the contents cached in $P$ are also available in $Q$ with the same or better streaming quality, {\it i.e.},
\begin{equation}\label{PQmonotonicity}
    s_u(P,i) \leq s_u(Q,i), \text{ for all } i\in \mathcal{K}. 
\end{equation}
This is easily proven by contradiction, assuming that there exist a content $\eta$ such that $s_u(P,\eta) > s_u(Q,\eta)$.

Next, by Definition \ref{def:Rvalue}, the following inequalities are true
\begin{eqnarray}
V_{ui}(P) &\leq& V_{ui}(Q), \\
V_{u[k]} (P) &\leq& V_{u[k]} (Q), \text{ for all } i,k\in \mathcal{K}. \label{VineqPQ}
\end{eqnarray}
Finally,   it follows by \eqref{F_u_formula} that $F_u^*(P) \leq F_u^*(Q)$.
\end{proof}

Next, we define the marginal gain of $F^*_u$ and we state an immediate consequence of Lemma \ref{monotone}.
\theoremstyle{plain} \newtheorem{monot_corollary}{Corollary} 
\begin{monot_corollary}[Marginal gain] \label{monot_corollary}
For a cache placement $P$, and a pair $(i,j)$ such that $(i,j) \not\in P$, we denote by  
\begin{equation*}
\Delta F_u^* (P,(i,j)):= F_u^*(P') - F_u^* (P), 
\end{equation*}
 where $P':=P \cup \{(i,j)\}$, the marginal gain of $F_u^*$ at $P$ with respect to $(i,j)$. Then, $ \Delta F_u^* \left( P,(i,j)\right) \geq 0$.\end{monot_corollary}

\theoremstyle{plain} \newtheorem{submodular}[rate_formula]{Lemma} 
\begin{submodular} \label{submodular} 
The set function $F_u^*$ is submodular\footnote{For definition, see \cite{Fisher1978}.}
for all $u\in \mathcal{U}$. 
\end{submodular}

We consider two placements $A$ and $B$ such that $A \subseteq B \subseteq \mathcal{P}$ and $(i,j)\in \mathcal{P}\setminus B.$ We need to prove that 
\begin{equation} \label{to_prove_submod}
\Delta F_u^*(A,(i,j)) \geq \Delta F_u^*(B,(i,j)).
\end{equation}
In other words, the marginal benefit of adding content $i$  to the cache $j$ in $A$ is greater than or equal to the marginal benefit in $B$. This means that the function $F_u^*$ has the diminishing returns property. 

In order to prove Lemma \ref{submodular}, we need a few intermediate results.  \emph{All the results and proofs given here are true for any $u\in \mathcal{U}$, so, for simplicity,  the index $u$ will be omitted  throughout the two following lemmas and the proof of Lemma~7.}

\theoremstyle{plain} \newtheorem{suiP}[rate_formula]{Lemma}
\begin{suiP} \label{suiP}
Given $P$, a placement of cached items,  and a pair $(i,j) \in (\mathcal{K} \times \{1,\ldots,C\})$ such that $(i,j) \not\in P$, we write $P':=P\cup (i,j)$. Then, $s(P',i)= max \{s_{j}, s(P,i) \}$.
\end{suiP}

\begin{proof}
By~\eqref{highest_rate}, after adding $(i,j)$ to the placement, if $s_j$ is greater than $s(P,i)$, then $i$ will be retrieved from cache $j$, {\it i.e.}, $s(P',i)=s_j$. On the other hand, if cache $j$ does not offer a better quality for the user than before, the quality associated to $i$ will stay the same, {\it i.e.}, $s(P',i)=s(P,i)$.
\end{proof}

\theoremstyle{plain} \newtheorem{Delta_f_lemma}[rate_formula]{Lemma}
\begin{Delta_f_lemma} \label{Delta_f_lemma}
Given $P$, a placement of cached items, and a pair $(i,j) \in (\mathcal{K} \times \{1,\ldots,C\})$ such that $(i,j) \not\in P$, the following statements are true:
\begin{enumerate}

\item[a.] $\Delta F^* (P,(i,j))= 0$  if and only if  $s_{j} \leq s(P,i)$; 
\item[b.]$\Delta F^* (P,(i,j))> 0$ if and only if $s_{j} >  s(P,i)$. 
\end{enumerate}
In the second case, the marginal gain is equal to
\begin{eqnarray*}
\Delta F^*  (P,(i,j))&=& (1-\alpha)\;p_{i}\;(s_{j}-s(P,i))\\ &+&  \begin{cases}
    0,&\text{if }\; V_{i}(P')\leq V_{[N]}(P),  \\
    V_{i}(P')-M, &\text{if }\; V_{i}(P')>V_{[N]}(P),   
    \end{cases} 
\end{eqnarray*}
where $P'=P\cup (i,j)$ and  $M= \max\{V_{i}(P), V_{[N]}(P) \}$.

\end{Delta_f_lemma}

Essentially, Lemma \ref{Delta_f_lemma} states that adding $(i,j)$ to $P$ will lead to a positive marginal gain of $F^*$ if and only if the cache $j$ can provide a higher (cache-spec.) quality for the  user than any other cache where $i$ was already cached~(in $P$).

\begin{proof}
By Lemma \ref{y_optimal}, $\Delta F^* (P,(i,j))> 0$ if and only if 
\begin{eqnarray}
& & \sum_{k=1}^{N}\left( V_{[k]} (P')- V_{[k]} (P) \right) >0 \label{Svineq}\\
&\text{ or }& \sum_{k\in \mathcal{K}} (s(P',k)-s(P,k))>0. \label{Ssineq}
\end{eqnarray}
Given that the only difference between $P$ and $P'$ is the content $i$ in cache $j$, the qualities of the contents other than $i$ remain the same as the addition of $(i,j)$ does not affect them. As a result,  the inequality in \eqref{Ssineq} is true if and only if $s(P',i)>s(P,i)$. By Lemma \ref{suiP}, $s(P',i)=s_j$ and, therefore, \eqref{Ssineq}  is equivalent to $s_j > s(P,i)$. Next, when the inequality \eqref{Svineq} holds, $V_i(P')>V_i(P)$ because, otherwise, $V_{[k]}(P')=V_{[k]}(P)$, for every $k\in \mathcal{K}$. By \eqref{def:eq:V}, the inequality $V_i(P')>V_i(P)$ is equivalent to $s(P',i)=s_j>s(P,i)$.
Hence, we proved that the inequality in \eqref{Svineq} implies $s_j > s(P,i)$, which is equivalent to \eqref{Ssineq}. Therefore, statement (b) holds. Since $\Delta F^* (P,(i,j)) \geq 0$ (by Corollary \ref{monot_corollary}), it follows that $\Delta F^* (P,(i,j))= 0$  if and only if  $s_{j} \leq s(P,i).$

Next, we  calculate $\Delta F^* (P,(i,j))$  when  $s_{j} > s(P,i)$.  First, note that, by \eqref{F*formula}, the expression $F^*(P')-F^*(P)$ consists of two summands. The summand with coefficient $(1-\alpha)\, p_i$ is equal to
$ \sum_{k\in \mathcal{K}} (s(P',k)-s(P,k)) = s_{j} - s(P,i)$.
In order to calculate the other summand, we compare $V_i(P')$ with the V-values of the recommended items in $P$ \emph{before} adding $(i,j)$, {\it i.e.}, the values $V_{[k]}(P)$ for $k=1,\ldots, N$.

If $V_i(P')< V_{[N]}(P)$, content $i$ will not feature in the recommendations list after caching it in $j$. If $V_i(P')= V_{[N]}(P)$, then content $i$ may make it to the recommendations list by replacing the $[N]$-th item in the list, assuming that ties are broken arbitrarily in the selection process. In both cases, nothing changes in terms of the $[N]$ highest V-values in $P'$, which implies that $\sum_{k=1}^N \left( V_{[k]} (P')- V_{[k]} (P) \right) = 0$.

On the other hand, if  $V_i(P') > V_{[N]}(P)$, content $i$ will definitely feature in the recommendations list after adding it in $j$, which implies \eqref{Svineq}. 
We consider  two subcases:
\begin{itemize}
    \item $i$ was already among the recommendations in $P$ even before caching it in $j$, {\it i.e.}, $V_{i}(P) \geq V_{[N]}(P)$. In this case, since the streaming quality is better at $j$, a part of the marginal gain will come from the difference in V-value of $i$ before and after adding $(i,j)$. This means that $\sum_{k=1}^N \left( V_{[k]} (P')- V_{[k]} (P) \right) = V_{i}(P') - V_{i}(P)$.
    \item $i$ was not recommended before caching it in $j$, {\it i.e.}, $V_{i}(P) < V_{[N]}(P)$. Since $V_i(P')> V_{[N]}(P) >V_{i}(P)  $, content $i$ gets into the recommendations list by replacing the $N$-th recommendation. Hence, a part of the marginal gain will come from the difference of the new V-value of $i$ and the V-value of the $[N]$-th  item in $P$, {\it i.e.}, $\sum_{k=1}^N \left( V_{[k]} (P')- V_{[k]} (P) \right) = V_{i}(P') - V_{[N]}(P)$.
\end{itemize}
Then, the result follows by replacing the findings above in the expression $F^*(P')-F^*(P)$. 
\end{proof}

We can now prove Lemma \ref{submodular}.

\begin{proof}[Proof of Lemma \ref{submodular}]
For two placements $A$ and $B$ such that $A \subseteq B \subseteq \mathcal{P}$ and a pair $(i,j)\in \mathcal{P}\setminus B$, we need to prove \eqref{to_prove_submod}.
As before, $A'$ and $B'$ are the sets $A\cup (i,j)$ and $B\cup (i,j)$ respectively. Since $A\subseteq B$, eq. \eqref{PQmonotonicity} (Lemma \ref{monotone}) implies that 
\begin{equation} \label{ABmonotonicity}
  s(A,i) \leq s(B,i).
\end{equation}

In line with Lemma \ref{Delta_f_lemma}, we examine the following cases: {\it i)} $ \Delta F^*(A,(i,j))=0$; {\it ii)} $ \Delta F^*(A,(i,j))>0$.
The first case is equivalent to $s_{j} \leq s(A,i)$, by Lemma \ref{Delta_f_lemma}. 
Then, by \eqref{ABmonotonicity}, $s_{j} \leq s(B,i) $. We invoke once again Lemma \ref{Delta_f_lemma} and we get that $\Delta F^*(B,(i,j))= \Delta F^*(A,(i,j))=0$.

Concerning the second case, it is equivalent to $s_{j} > s(A,i)$ and $\Delta F^*(A,(i,j))$ is given by the formula in Lemma \ref{Delta_f_lemma}. We consider three subcases: 
\begin{itemize}
    \item $s_{j} \leq s(B,i) $;
    \item $s_{j} > s(B,i) $ and $V_i(B') \leq V_{[N]}(B)$;
     \item $s_{j} > s(B,i) $ and $V_i(B') > V_{[N]}(B)$.
\end{itemize}

In the first subcase, $ \Delta F^*(B,(i,j))=0$, by Lemma \ref{Delta_f_lemma} and, therefore, $\Delta F^*(A,(i,j)) > \Delta F^*(B,(i,j))=0$.

Next, $s_{j} > s(B,i) $ is equivalent to $\Delta F^*(B,(i,j))>0$.
Since $s_{j} > s(A,i) $ as well, it holds that
\begin{equation} \label{sAandsB}
    s_j= s(A',i)= s(B',i).
\end{equation}
If $V_i(B') \leq V_{[N]}(B)$, Lemma \ref{Delta_f_lemma},  \eqref{ABmonotonicity} and \eqref{sAandsB} imply that 
\begin{equation*}
\Delta F^*  (B,(i,j))= (1-\alpha)\, p_{i}\,(s_{j}-s(B,i)) \leq \Delta F^*  (A,(i,j)).
\end{equation*}
If $V_i(B') > V_{[N]}(B)$, by \eqref{ABmonotonicity}, \eqref{sAandsB} and \eqref{VineqPQ}, it follows that
\begin{equation} \label{last_one}
    V_i(A')=V_i(B')>V_{[N]}(B) \geq V_{[N]}(A).
\end{equation}
Combining this with \eqref{ABmonotonicity} and Lemma \ref{Delta_f_lemma}, in order to prove $\Delta F^*  (A,(i,j)) \geq  \Delta F^*  (B,(i,j))$, we only need to prove that
\begin{equation} \label{max:to:prove}
    \max\{V_{i}(A), V_{[N]}(A) \} \leq \max\{V_{i}(B), V_{[N]}(B) \}.
\end{equation}

It follows by \eqref{ABmonotonicity} that $V_i(A) \leq V_i(B)$, and therefore $ V_i(A) \leq \max\{V_{i}(B), V_{[N]}(B) \}$. Moreover, $ V_{[N]}(A) \leq \max\{V_{i}(B), V_{[N]}(B) \}$, by \eqref{last_one}.
We then obtain \eqref{max:to:prove} and this concludes the proof.
\end{proof}

\theoremstyle{plain} \newtheorem{submodular_linear_combi}[rate_formula]{Lemma} 
\begin{submodular_linear_combi} \label{submodular_linear_combi}
The set function $f^*$, as defined in \eqref{def_f*X}, 
is monotone increasing and submodular. 
\end{submodular_linear_combi}

\begin{proof}
By Lemma \ref{lemma_max_y}, $f^*(X)= \sum_{u\in \mathcal{U}} F_u^*(X)$. It is easy to prove that monotonicity and submodularity are preserved under non-negative linear combinations.
Therefore, the result is an immediate consequence of Lemmas \ref{monotone} and \ref{submodular}.
\end{proof}

\subsection{MoSE algorithms and guarantees} \label{sub:guarantees}

We managed to prove through the decomposition in \eqref{eq:decompos} that $f^*(X)$ is submodular. The theory on submodularity optimization suggests that different greedy algorithm variants give constant approximations for the outer problem, and thus for the MoSE problem. In fact, the factor of approximation depends on the constraints in~\eqref{matroid_constraint}.

\subsubsection{The case of equal-sized contents}

We define a greedy algorithm that we call the  \emph{MoSE algorithm}. This algorithm  starts with a placement $P$ consisting of empty caches (except for the large cache that contains the entire catalogue) and greedily fills one by one all the available shots. In every round of selection, it calculates the marginal gain of $f^*$ at $P$ with respect to at most $C \cdot |\mathcal{K}|$ elements, {\it i.e.}, pairs (content, cache), by solving the {I}nner {A}lgorithm (as subroutine). It then selects and adds to $P$ the element that maximizes the marginal gain (ties broken arbitrarily), before the next selection round begins. The algorithm is summarized in Algorithm 2.

\begin{algorithm}[h]
\NoCaptionOfAlgo

\SetAlgoLined
\SetKwInOut{Input}{Input}
\SetKw{Return}{Return}
\Input{$C$, $\{\mathcal{C}_j \}, \mathcal{U}$, $\mathcal{K}$, $\{N_u\}$, $\{s_{uj}\}, \{r_{ui}\}$, $\{\beta_u\}$, $\{\alpha_u\}$}

 Start with empty caches, {\it i.e.}, $P= \cup_{j=1}^C  P^{(j)}$, where  $P^{(j)}=\emptyset$, for all $j=1,\ldots,C $\\
 \textbf{Outer algorithm:}\\
 \While{caches are not full, {\it i.e.}, $|P^{(j)}| < \mathcal{C}_j$ for all $j$,}{
 \For{every (not full) cache $j=1,\ldots,C$,}{
  \For{every content $i\in \mathcal{K}$ s.t. $(i,j)\notin P^{(j)}$,}{
 Estimate $\Delta f^* \left( P, (i,j) \right)$ by calling \textbf{Inner Algorithm(X)}; 
 Store $\max \Delta f^* \left( P, (i,j) \right)$.
 }
  }
 $(\eta,\theta):= \text{ argmax}_{(i,j)} \Delta f^* \left( P, (i,j) \right)$.\\
Add $(\eta,\theta)$ to $P$, {\it i.e.},
$P^{(\theta)} \leftarrow P^{(\theta)}\cup (\eta,\theta)$. \\
 }

\Return{$X^* \leftrightarrow P, Y^*= f^*(X^*)$} 
 \caption{Algorithm 2: \textbf{MoSE algorithm (equal-sized contents)}}
\end{algorithm}

 \theoremstyle{definition} \newtheorem{main_result}{Theorem}
\begin{main_result}[Homogeneous sizes] \label{main_result}
  If we let $OPT$ denote the optimal objective function value of the MoSE problem with equal-sized contents, and $(X^*, Y^*)$ denote the feasible solution given by the MoSE algorithm, then \begin{equation*}
f(X^*,Y^*) \geq \frac{1}{2}  OPT.
\end{equation*}
\end{main_result}
\begin{proof}
Since the constraints in \eqref{matroid_constraint} are matroid constraints, as in \cite{femto_JOURNAL2013}, the theory on submodular maximization~\cite{Fisher1978} suggests that
 a $1/2$-approximation is achievable by the above greedy algorithm.
\end{proof}

\subsubsection{The general case of contents of heterogeneous sizes} \label{subsub_size}

The fundamental difference between the two cases is the capacity constraints.  The constraints in \eqref{matroid_constraint} in the general case are knapsack constraints. However, the MoSE algorithm is oblivious of the content's size. The following algorithm  is an adaptation of the MoSE algorithm that takes size into account. More precisely, in every round of selection, it adds to the cache   the element (content, cache) that maximizes the ratio of marginal gain to the content's size, while satisfying the constraints in \eqref{matroid_constraint}. It is summarized in Algorithm 3.

\begin{algorithm}[h]
\NoCaptionOfAlgo

\SetAlgoLined
\SetKwInOut{Input}{Input}
\SetKw{Return}{Return}
\Input{Same as in MoSE alg. and $\{\sigma_i\}$} 

 Start with $P= \cup_{j=1}^C  P^{(j)}$, where  $P^{(j)}=\emptyset$, for all $j$\;
  \textbf{Outer algorithm:}\\
  \While{caches are not full, {\it i.e.},$\sum_{k\in P^{(j)}} \sigma_k < \mathcal{C}_j$,}{
 \For{every (not full) cache $j=1,\ldots,C$,}{
 \For{every content $i\in \mathcal{K}$ such that $(i,j)\notin P^{(j)}$ and $\sigma_i \leq \mathcal{C}_j-\sum_{k\in P^{(j)}}\sigma_k$,}{
 Estimate $\delta  f^* ( P, (i,j)) := \frac{ \Delta f^* \left( P, (i,j) \right)}{\sigma_i}$\\
 by calling \textbf{Inner Algorithm}(X); \\
 Store $\max_{(i,j)} \delta f^* \left( P, (i,j) \right)$.
 }
  }
$(\eta ,\theta):= \text{ argmax}_{(i,j)} \delta f^* \left( P, (i,j) \right)$. Add it to $P$.\\

 }
\Return{$X^* \leftrightarrow P, Y^*= f^*(X^*)$} 
 \caption{Algorithm 3: \textbf{s-MoSE algorithm (size-aware)}}
\end{algorithm}

 \theoremstyle{definition} \newtheorem{gen_result}[main_result]{Theorem}
\begin{gen_result}[Heterogeneous sizes]
If we let $OPT_s$ denote the optimal objective function value of the MoSE problem in the general case (contents of heterogeneous sizes), and $(X^*, Y^*)$, $(X_s, Y_s)$ denote the feasible solutions given by the MoSE and s-MoSE algorithms respectively, then
\begin{equation*}
    \max \{ f(X^*, Y^*), f(X_s,Y_s)\} \geq \frac{1-1/e}{2} OPT_s.
\end{equation*}
\end{gen_result}

\begin{proof}
In the case of variable-sized contents, both MoSE and s-MoSE algorithms can perform arbitrarily badly~\cite{leskovec-max-submodular-knapsack}. According to the result in~\cite{leskovec-max-submodular-knapsack}, it suffices to choose the maximum objective function value achieved by the two algorithms in order to achieve a $\frac{1-1/e}{2}$-approximation. 
\end{proof}


\subsubsection{Complexity, implementation speed-ups and distributed techniques} \label{subsec:complexity}
It is easy to see that the complexity of both the MoSE and the s-MoSE algorithms is the same. The algorithms need to run  at most $\sum_{j=1}^{j=C} |\mathcal{C}_j |$ times in order to fill all caches. At each iteration, they evaluate the marginal gain of at most $C \cdot|\mathcal{K}|$ pairs (content, cache).  For every evaluation, they call the \emph{Inner Algorithm} of complexity $O (|\mathcal{U}| \cdot  |\mathcal{K}|)$ and complete $|\mathcal{U}|$ calculations that concern the non-recommendation part of the objective function. 
Therefore, the total complexity of the MoSE and s-MoSE algorithms is
$ O ( |\mathcal{U}| \cdot  |\mathcal{K}|^2 \cdot C \cdot \sum_{j=1}^{j=C} | \mathcal{C}_j |)$.

Implementation-wise, there is a way to speed up both algorithms by using the so-called  lazy evaluations method~\cite{leskovec-max-submodular-knapsack}. This method takes advantage of the monotonicity and submodularity of the objective function in order to avoid unnecessary calculations in the selection process of the caching placement. Recent  works propose methods for further acceleration, \emph{e.g.,} randomized greedy algorithm in~\cite{mirzasoleiman2015lazier}.

Finally, we note that there is a technique suggested in the literature for distributed/multi-processor implementations of greedy algorithms of submodular maximization~\cite{mirzasoleiman2016distributed}, such as the MoSE algorithm. More precisely, for a set of $m$ processors/nodes, this technique starts by  partitioning the ground set, \emph{i.e.,} the catalogue of contents, into $m$ subsets. Then each processor solves \emph{in parallel} the MoSE problem only on one of the subsets by applying our proposed policy. This leads to $m$ solutions~(caching allocation and users' recommendations) $DS_1,\ldots, DS_m$. Next, we define a new subset of the catalogue by merging the contents for which the caching variable was equal to $1$ in at least one of the previous solutions. We then run on this subset the MoSE algorithm which gives the solution $MS$. Finally, among the solutions $DS_1,\ldots, DS_m$ and $MS$, the one with the largest value of the objective function is selected. We note that, under some conditions, this implementation  offers approximation guarantees that depend on the guarantees of the centralized algorithm and on the number $m$. For more details on the algorithm and these approximation guarantees  we refer the reader to~\cite{mirzasoleiman2016distributed}. We will call \emph{m-DMoSE} the algorithm described above, where $m$ is the number of processors.
 
\subsection{The single-cache case $(C=1)$} \label{caseC1}

We study now the case where $C=1$, {\it i.e.}, apart from the large cache $C_0$, there is only one cache.  We prove that, in this case,  the  MoSE problem can be transformed  into an Integer Linear Program~(ILP) problem and, thus, common optimization methods can be applied to find the optimal solution  for small problem's instances. This will be useful in the next section since it will allow us to compare the performance of our algorithm with the optimal joint policy.


We introduce the variable $\{z_{ui}\}_{i,u}$ such that $z_{ui}=x_i y_{ui}$. The objective of the MoSE problem in \eqref{objective} becomes
\begin{eqnarray}
g(X,Y,Z)= \sum_{u\in \mathcal{U}} \sum_{i \in \mathcal{K}} \big[ \dfrac{\alpha_u}{N} \left( \left( s_{u1}-s_{u0} \right) z_{ui} + s_{u0} y_{ui} \right) \nonumber \\ 
+ \left(1-\alpha_u \right) p_{ui} \left( ( s_{u1} - s_{u0} )x_i + s_{u0} \right) + \beta_u y_{ui} \varphi(r_{ui})  \big].
\end{eqnarray}
Therefore, the MoSE problem for $C=1$ is equivalent to:

\theoremstyle{definition} \newtheorem*{zproblem}{Z Problem}
\begin{zproblem} 
\begin{eqnarray}
&\underset{X,Y,Z}{\text{maximize  }}& g(X,Y,Z) \label{obj_g}\\
&\text{subject to } &  \eqref{matroid_constraint}, \eqref{N_recomm},  \nonumber\\
&&  z_{ui}=x_i y_{ui}; \label{z-x-y-nonlinear-constr} \\
&& x_{ij}, y_{ui}, z_{ui} \in \{0,1\}. \label{binary_xyz}
\end{eqnarray}
\end{zproblem}

The equivalence comes from the fact that a pair $(\tilde{X}, \tilde{Y})$, where $\tilde{X}= \{\tilde{x}\}_i$ and $\tilde{Y}=\{\tilde{y}\}_{u,i}$, is optimal for the MoSE problem if and only if $(\tilde{X}, \tilde{Y}, \tilde{Z})$ where $\tilde{Z}= \{ \tilde{z} \}_{ui}$ such that $\tilde{z}_{ui}= \tilde{x}_i \tilde{y}_{ui}$ is optimal for the Z problem.

Notice that, although $g(X,Y,Z)$ is linear in the variables $X,Y$ and $Z$, the constraints  \eqref{z-x-y-nonlinear-constr} are nonlinear. However, we will prove that these constraints can be replaced by the following inequalities:
\begin{eqnarray}
z_{ui}&\leq & x_i, \label{inequality1-z}\\
z_{ui} &\leq & y_{ui}, \text{ for all } u\in \mathcal{U}, i\in \mathcal{K}.  \label{inequality2-z}
\end{eqnarray} 

\theoremstyle{plain} \newtheorem{equivalence_ILP}[rate_formula]{Lemma} 
\begin{equivalence_ILP} \label{equivalence_ILP}
The MoSE problem for $C=1$ is equivalent to the following ILP problem:
\end{equivalence_ILP}

\theoremstyle{definition} \newtheorem*{QoE_ILP}{MoSE ILP problem}
\begin{QoE_ILP} 
\begin{eqnarray*}
&\underset{X,Y,Z}{\text{maximize  }}& g(X,Y,Z) \\
&\text{subject to }& \eqref{matroid_constraint}, \eqref{N_recomm}, \eqref{binary_xyz} - \eqref{inequality2-z}.
\end{eqnarray*}
\end{QoE_ILP}

\begin{proof}
It suffices to prove that a solution for the Z problem is also a solution for the MoSE ILP problem and the inverse. 
Let us assume that $(\bar{X}, \bar{Y}, \bar{Z})$ is a solution for the Z problem. Since $\bar{X}$ and $\bar{Y}$ are binary variables, the expression $\bar{z}_{ui}= \bar{x}_i \bar{y}_{ui}$ implies the inequalities \eqref{inequality1-z} and \eqref{inequality2-z}. Hence, $(\bar{X}, \bar{Y}, \bar{Z})$ is also a solution for the MoSE ILP problem.

Inversely, let us assume that $(\tilde{X}, \tilde{Y}, \tilde{Z})$ is a solution for the MoSE ILP problem. It suffices to prove that $\tilde{z}_{ui}= \tilde{x}_i \tilde{y}_{ui}$ for every $u\in \mathcal{U}$ and $i\in \mathcal{K}$.
For the $u$ and $i$ such that $\tilde{x}_i=0$ or $\tilde{y}_{ui}=0$, the inequality constraints imply that $\tilde{z}_{ui}=0$. For the $u$ and $i$ such that $\tilde{x}_i=1$ and $\tilde{y}_{ui}=1$, we will necessarily have that $\tilde{z}_{ui}=1$ since the coefficient of $z_{ui}$ in the objective function $g$ is strictly positive in a maximization problem. Hence, considering that all variables are binary, it follows that  $\tilde{z}_{ui}= \tilde{x}_i \tilde{y}_{ui}$, and this concludes the proof.
\end{proof}

\section{Performance Evaluation}
\label{sec:perf}
In this section, we  validate the theoretical approximation guarantees of the proposed policy (\emph{MoSE algorithm}) and we compare it with other policies in a variety of scenarios.
\subsection{Scenario 1}
As a first step, we compare the performance of the MoSE algorithm with its distributed implementations 2-DMoSE and 4-DMoSE~(\emph{i.e.,} in 2 and 4 processors, see Sec.~\ref{subsec:complexity}) and with the optimal policy~(oracle). We consider a scenario with a single  cache and the large cache $C_0$ that contains the entire catalogue. As shown in Sec.~\ref{caseC1}, the MoSE problem for $C=1$ can be transformed into an ILP problem.  Therefore, in order to find the optimal policy~(oracle), we use the standard MATLAB solver which employs methods such as branch-and-bound, cutting-plane method or exhaustive search.

We consider $20$ users connected to the cache and a catalogue of $200$ unit-sized contents. We assume that the cache can fit $15$ contents and every user receives $N=2$ recommendations. The small size of the scenario is necessary to be able to calculate the optimal objective value. We will consider much larger scenarios subsequently.  Moreover, the impact of the recommendations is determined by $\alpha_u$, whose values follow a uniform distribution between $0.7$ and $0.9$ (in line with the statistics gathered on Netflix \cite{gomez2016netflix}).
In this scenario, we consider a synthetic dataset for the utilities $r_{ui}$ and the popularities $p_{ui}$. We chose $p_{ui}$ such that the aggregate content popularities over all users, {\it i.e.}, $\sum_u p_{ui}$, follow a Zipf distribution (with parameter $0.6$). 
Then, $r_{ui}$ are chosen randomly in $[0,1]$ such that their normalized value, {\it i.e.}, $r_{ui}/\sum_k r_{uk}$, are equal to  $p_{ui}$, for every $i\in \mathcal{K}$, as in \eqref{popularities_r}.

In this scenario, we  measure the SQ as cache hits with the values $s_{uj}$ as in~\eqref{psi_function_hit}, and the  RQ~(Def. \ref{def_quality}) by considering $\varphi(r_{ui})=\log(r_{ui})$. For a variety of values  of $\beta_u=\beta>0$, we queried the oracle and we calculated the MoSE given by the proposed algorithm and its distributed implementations. For some of the values $\beta$, Table~\ref{table_approx} shows the approximation ratio achieved and Table~\ref{table_execution} shows the average execution time per instance of the problem.

As we saw in Sec.~\ref{sub:guarantees}, the ratio $f(X^*,Y^*)/ OPT$ cannot be lower than $1/2$.  We observe that, in practice, the achieved ratio is much higher than $1/2$, as is also observed for other submodular problems, {\it e.g.}, in~\cite{dehghan2016complexity}. In fact, among all the different values of $\beta$ we considered~(30 in total), the lowest observed approximation ratio was equal to $0.9757$. Moreover, the approximation ratios achieved by the distributed algorithms 2-DMoSE and 4-DMoSE are also close to $1$.
However, as expected by the discussion in Sec.~\ref{subsec:complexity}, they do not perform as well as the (centralized) MoSE algorithm for some instances of the problem~(\emph{e.g.,} for $\beta=3.2$).

\theoremstyle{definition} \newtheorem{obs_optimal}[]{Observation} 
\begin{obs_optimal} 
Our numerical results validate the theoretical approximation guarantees  of our policy and also suggest a much better approximation ratio in practice.
\end{obs_optimal}

\begin{table}[!th] 
\centering 

\caption{Approximation ratio $(f(X^*,Y^*)/OPT)$} 
\begin{tabular}{|c|c|c|c|c|}
\hline
\textbf{Parameter $\beta$} & $0.01$ & $1$ & $1.7$ & $3.2$  \\
\hline
\textbf{Approx. ratio for MoSE${}^*$} & $1$ & $0.9977$ & $0.9979$ & $1$ \\
\hline
\textbf{Approx. ratio for 2-DMoSE} & $0.9998$ & $0.9935$ & $0.9979$ & $0.9818$  \\
\hline
\textbf{Approx. ratio for 4-DMoSE} & $0.9998$ & $0.9750$ & $0.9979$ & $0.8836$  \\
\hline
\multicolumn{5}{l}{ \small{${}^*$}\scriptsize{theoretical lower bound: $0.5$ (see Theorem~\ref{main_result})}}  
\end{tabular} \label{table_approx}

\end{table}

\begin{table}[!th] 
\centering 
\caption{Execution Time (AVG) Per Policy } 
\begin{tabular}{|c|c|c|c|}
\hline
\textbf{MoSE} & \textbf{2-DMoSE} &  \textbf{4-DMoSE} &\textbf{Oracle} \\
\hline
0.0221 sec. & 0.0147 sec. & 0.0111 sec. & 120.195 sec. \\
\hline
\end{tabular} \label{table_execution}
\end{table}

Regarding Table~\ref{table_execution}, we see that the average execution time of the MoSE algorithm is of much lesser magnitude than the oracle's one. We note that we implemented MoSE with the  lazy evaluations technique which avoids unnecessary calculations~(see Sec.~\ref{subsec:complexity}).
When MoSE is implemented in 2 and 4 processors, the execution time decreases significantly, and we observe a 2x speedup.

\theoremstyle{definition} \newtheorem{obs_execution}[obs_optimal]{Observation} 
\begin{obs_execution}
Implementing the MoSE algorithm leads to significant savings in execution time when compared with the oracle. These savings can be further pronounced in the case of a distributed/multi-processor implementation.
\end{obs_execution}

Next, we investigate if this close-to-optimal performance is reflected in the SQ-RQ tradeoffs. At the same time, we will compare these tradeoffs with the ones achieved by a proposed heuristic in the literature for a similar problem~\cite{chatzieleftheriou2019joint-journal}. 

\theoremstyle{remark} \newtheorem*{CAwR}{Cache-aware recommendations (CAwR)}
\begin{CAwR} 
CAwR~\cite{chatzieleftheriou2019joint-journal} makes caching and recommendation decisions at every cache independently. 
It decomposes the problem into the caching and recommendation steps.
First, given the content preference distribution for every user (equivalent to the content popularity distribution $p_{ui}$ or content utilities $r_{ui}$ of our model) and the weight every user gives to recommendations (the $\alpha_u$ of our model), the aggregate request probability of every content is calculated. Then, the $N$ items with the highest probability are cached. Note that, in the case of variable-sized contents, the cache allocation decisions are made by solving a $0-1$ knapsack problem, where the ``value" of every content is the aforementioned probability and the ``weight'' is its size. Then, in the recommendation step, the recommendations are made partially by cached contents and by non-cached contents that are of high utility for the particular user. The balance between cached and non-cached contents is determined by a so-called distortion parameter $r_d\in [0,1)$, which is similar to the parameter $\beta$ of our model.
\end{CAwR}

\begin{figure}[htp]
\centering
\includegraphics[width=6cm, trim={0.3cm 0cm 0.43cm 0.4cm},clip]{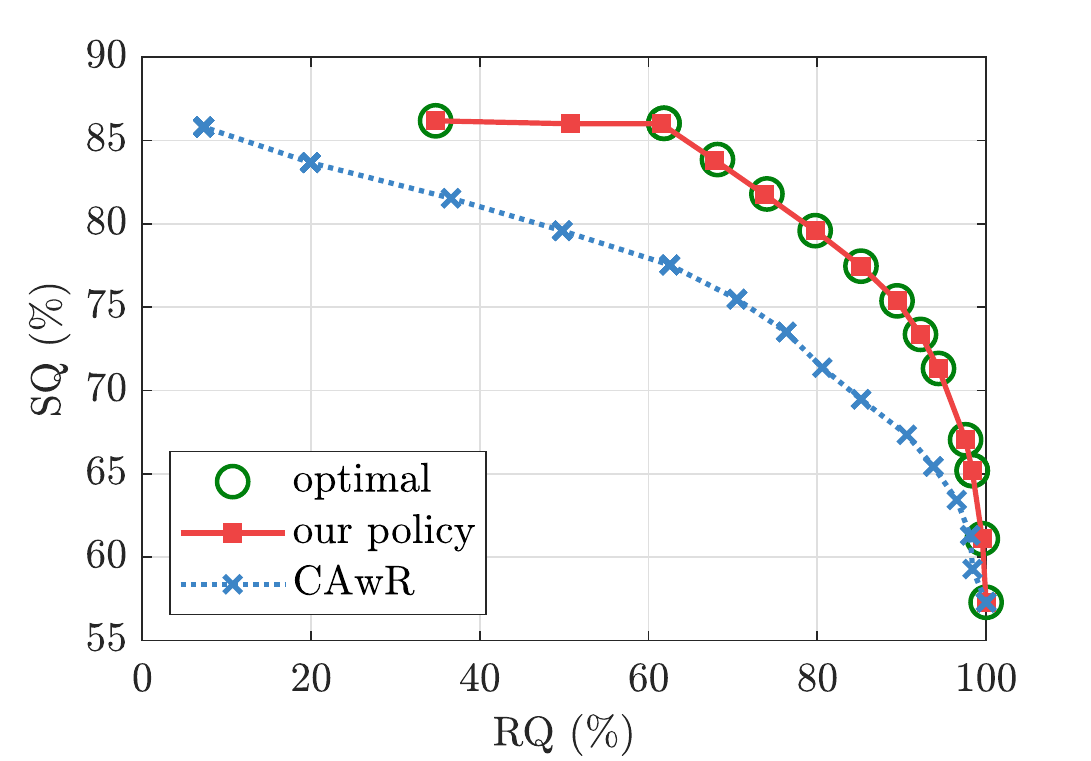}

\caption{Scenario 1, SQ-RQ tradeoff points for some values of the parameters $\beta$ and $r_d$.}
\label{tradeoff_synth_optimal}
\end{figure}

\begin{figure*}[t]
\centering
\centering
\subfloat[Equal-sized contents]{\includegraphics[width=6cm,trim={0.12cm 0cm 0.55cm 0.1 cm},clip]{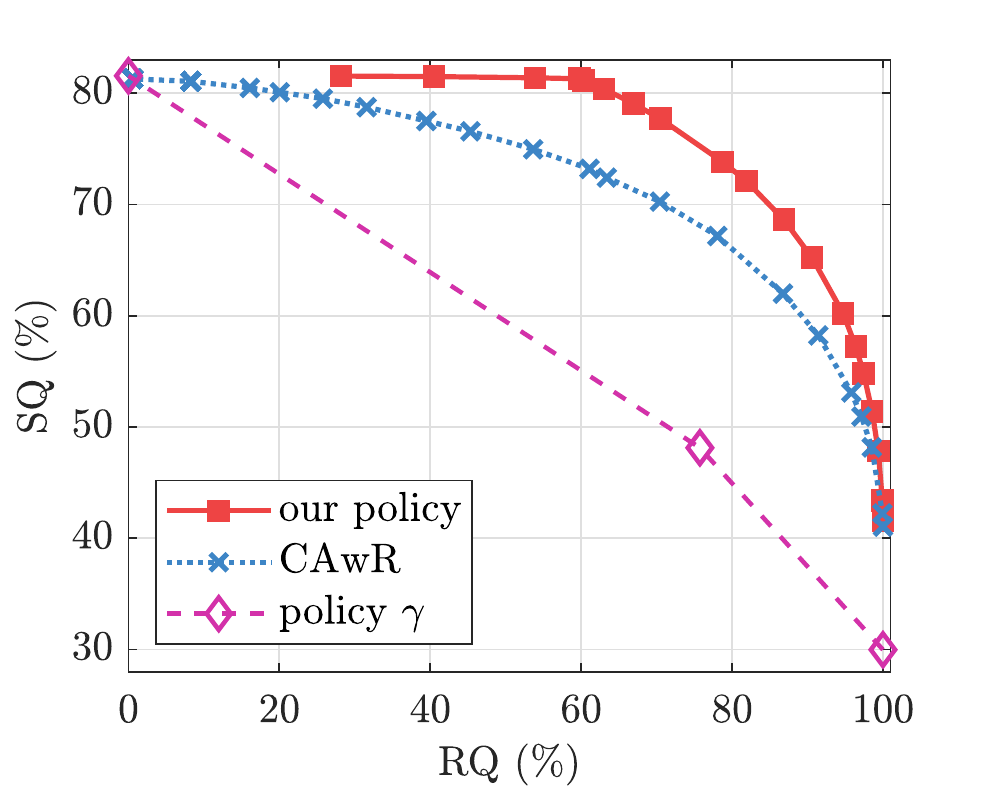}}
\hfil
\subfloat[Variable-sized contents]{\includegraphics[width=6cm,trim={0.12cm 0cm 0.55cm 0.1 cm},clip]{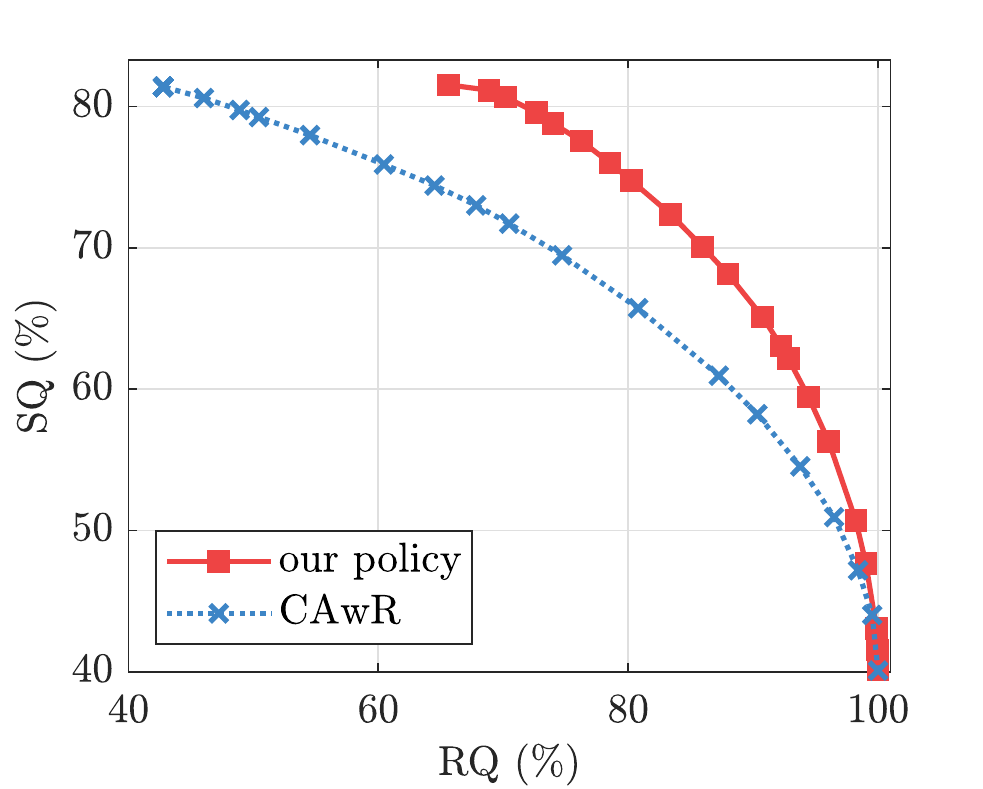}} 
\subfloat[Equal-sized contents, $r_{min}=0$ and $0.6$]{\includegraphics[width=6cm,trim={0.12cm 0cm 0.55cm 0.1 cm},clip]{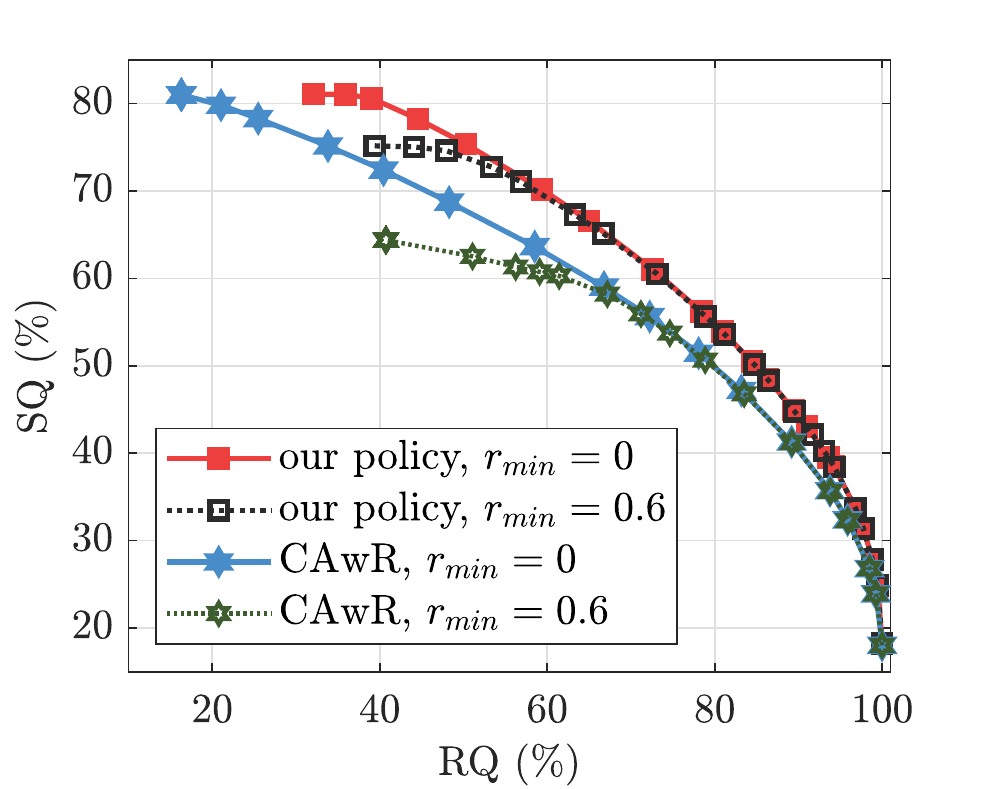}} 
\caption{Scenario 2, SQ-RQ tradeoff points. Comparison of our policy with the  policy CAwR proposed in~\cite{chatzieleftheriou2019joint-journal} and the baseline policy $\gamma$.} \label{trio1}
\end{figure*}

Figure \ref{tradeoff_synth_optimal} depicts the SQ-RQ tradeoffs given by the oracle, our policy, and CAwR as points in the plane. We obtained these tradeoffs for 30 different values of $\beta_u=\beta$ in the range $[0.01, 70]$ and the distortion parameter $r_d$. 
 The RQ values (x-axis) are normalized with respect to the two ``extreme'' policies A and C (defined in Section \ref{subsec:joint}). For example, $\text{RQ}=50\%$ implies that the  RQ value lies in the middle of the interval $[R_A, R_C]$, where $R_A$ and $R_C$ are the RQ values achieved by policies A and C respectively. Moreover, since $s_{uj}$ are as in \eqref{psi_function_hit}, the normalized SQ values~(y-axis) give the cache hit rate. 
 
 We remind the reader that each of these points corresponds to a different objective tradeoff, between SQ and RQ, that a CP might have, {\it i.e.}, these curves could also be interpreted as Pareto curves. As we discussed in Sec.~\ref{subsec:QoE}, $\beta$ captures the weight we attach to the RQ compared to the SQ. For a small value of $\beta$~(on the left), caching decisions are made based on the aggregate~(over all users) interest in contents and recommendations concern mostly cached items. This leads to high SQ/cache hits but compromised RQ.  As $\beta$ increases, we trade off a better RQ for a worse SQ. In fact, a better RQ would imply recommendations to each user that are close to her tastes and it is $\beta$ that  determines how close. As these tastes can differ from one user to the other and the cache capacity is limited, caches cannot store all the different recommended contents and this leads to decreased SQ/cache hits.

\theoremstyle{definition} \newtheorem{obs_Lsynth}[obs_optimal]{Observation} 
\begin{obs_Lsynth} 
Our policy's tradeoff curve almost coincides with the optimal. Furthermore, it dominates the tradeoff curve of CAwR, {\it i.e.,} our policy outperforms CAwR in terms of at least SQ or RQ (or both).
\end{obs_Lsynth}

For example, for a desired value of SQ of around $84\%$, CaWR achieves $20\%$ RQ and our policy $68\%$. More importantly, most of the tradeoffs of our policy ({\it e.g.}, around $80-95\%$ RQ and $70-80\%$ SQ) are not achievable by any tuning of the CAwR algorithm.  Finally, we observe that for large  $\beta$ and small $r_d$~(points in the extreme right) the recommendation and caching decisions of the two policies coincide. In fact, both policies recommend to each user the contents with the highest utility for the user and both policies store the contents with the highest aggregate probability to be requested~(given the aforementioned recommendations).

\subsection{Scenario 2}
We proceed with simulating larger scenarios. For this, we consider a single cache with $100$ or $200$ connected users and a catalogue consisting of $6000$ or $10000$  contents\footnote{Note that according to \cite{netflix_museum}, the total number of titles (movies and TV shows) available on Netflix in the USA is equal to $5848$.}. We consider realistic values (according to footnote \ref{footnote-cache-capacity}, p.~\pageref{footnote-cache-capacity}) for cache capacity varying from $1\%$ to $2.3 \%$ of the entire catalogue. The probabilities $\alpha_u$ are chosen randomly in $[0.7,0.9]$, in line with the statistics gathered on Netflix \cite{gomez2016netflix}, and $N$ varying from $2$ to $10$. For these experiments, we use a real dataset for the matrix of utilities $r_{ui}$:

\theoremstyle{remark} \newtheorem*{movie-data}{MovieLens dataset}
\begin{movie-data} 
The MovieLens dataset \cite{movielens-related-dataset} is a collection of $5$-star movie ratings collected on MovieLens, an online movie recommendation service. This dataset has also been  used  in  related works on caching and recommendations, {\it e.g.}, in~\cite{sermpezis2018sch-jsac}. Here, we used a variety of subsets of the total $20000263$ ratings available in the original dataset. It is commonly assumed that the utility of a content for a user is the \emph{predicted} rating of this user for the content \cite{amatriain_building}. Therefore, we interpret the rating as the content utility. Since the range of ratings is $0.5-5$ with $0.5$ increments, we map every rating $r$ to a random number in the interval $(r/5-0.1, r/5]$.
As is common, this matrix is quite sparse. To obtain the missing  ratings, we perform matrix completion through the TFOCS software \cite{TFOCS-paper}. TFOCS performs nuclear norm minimization in order to find the missing entries of a low-rank matrix.
\end{movie-data}

\subsubsection{Equal-sized contents}
We assume that the contents are of unit size. We will show the performance improvement achieved by our policy over a baseline scheme,  policy $\gamma$, and the earlier introduced CAwR. To begin, we define policy $\gamma$:

\theoremstyle{remark} \newtheorem*{policy-gamma}{Baseline policy $\gamma$}
\begin{policy-gamma} 
 It is a generalization of policies A and C. Policy $\gamma$ caches the most popular contents and then recommends a combination of cached contents and contents with high utility per user depending on the parameter $\gamma$. More specifically, it recommends $\lceil \gamma \cdot N \rceil$ cached contents, where $\lceil \cdot \rceil$ denotes the ceiling function, while the rest of the recommendations are the contents with the highest utility per user. For $\gamma=0$, policy $\gamma$ coincides with policy C, and, for $\gamma=1$, it coincides with policy A. 
\end{policy-gamma}

As before,  we  measure the SQ as cache hits and the  RQ as $\sum_i \log(r_{ui})$. In Figure \ref{trio1}(a), we plot the tradeoffs achieved by policy $\gamma$, CAwR, and our policy for different values of the parameters $\gamma$, $r_d$, and $\beta$  respectively. In this instance, $N=2$, which results in $3$ possible objective values for policy $\gamma$. In fact, one point corresponds to recommending $2$ cached items, the next one to recommending one among the cached items and the one with the highest utility, and the last  point to recommending the $2$ first contents ranked in terms of utility. 
\theoremstyle{definition} \newtheorem{obs_Lmovielens}[obs_optimal]{Observation} 
\begin{obs_Lmovielens} 
The SQ-RQ tradeoff curve of our policy dominates that of CAwR and that of the baseline policy $\gamma$ in large, realistic scenarios, driven by real datasets.
\end{obs_Lmovielens}

We notice, for example, that, in terms of SQ, there is a relative improvement of up to $10\%$ with respect to CAwR and of up to $54\%$ with respect to policy $\gamma$, while the improvement is much larger in terms of RQ. This is an encouraging finding that suggests that the theoretical gains could also be experienced in practice. Finally, note that the performance gain of our policy over  policy $\gamma$ is mainly
due to the joint decisions on caching and recommendations that our policy makes.

\subsubsection{Contents with heterogeneous sizes} So far, we have considered scenarios with equal-sized content ({\it e.g.}, chunks), as is often assumed in related work~\cite{femto_JOURNAL2013}. Here, we turn our attention to a scenario with contents of heterogeneous size, as analyzed in Section \ref{subsub_size}. The sizes of the contents were chosen in  $\{1,15\}$ and,  according to the findings of~\cite{workload_youtube} on YouTube videos, $90\%$ of the contents have a size of at most $2$ size units, while only $0.1\%$ have a size over $10$ size units. We adjust the cache capacity to $2.3\%$ of the total size of the catalogue  $\sum_{i\in \mathcal{K}} s_i$.
Figure \ref{trio1}(b) depicts the tradeoffs achieved by the two policies. In this context, our policy runs both the MoSE and s-MoSE algorithms and selects the maximum achieved objective function value between the two, as explained in Section~\ref{sub:guarantees}.  As expected, the difference between the tradeoff curves is similar to the one in Fig. \ref{trio1}(a). More specifically, we observe a relative gain of up to $63\%$ in RQ  and up to $15\%$ in SQ of our policy with respect to CAwR.

\theoremstyle{definition} \newtheorem{obs_size}[obs_optimal]{Observation} 
\begin{obs_size} 
Heterogeneous content sizes do not have an impact on the performance gains of our policy which, in this context, still outperforms existing schemes.
\end{obs_size}

\subsubsection{RQ-related constraints}
While the previous results are  promising, one might argue that the proposed policy could still recommend some rather unrelated contents, {\it i.e.,} contents of utility $r_{ui}$ close to $0$,  
in favor of a higher objective value, or worse, that some users might receive much better recommendations, {\it i.e.,} tailored to their tastes, than others. 
For this reason, we will  evaluate the performance of our policy and that of existing schemes when additional constraints on RQ are added to the problem. In particular, we measure RQ by considering $\varphi$ as in~\eqref{phi_rmin}. This leads to recommendations of contents whose utility per user is at least $r_{min}$. Since the recommendation decisions are made by solving the ``inner problem '' (as explained in Section \ref{effic_algor}), the caching decisions  also take into account this constraint.
Subsequently, we adjust CAwR such that, at the recommendation step, the contents  with $r_{ui}<r_{min}$ cannot be recommended to the user.

\begin{figure*}[t]
\centering
\centering
\subfloat[ $\psi$ linear, $\varphi$ logarithmic]{\includegraphics[width=6cm,trim={0.07cm 0cm 0.55cm 0.1 cm},clip]{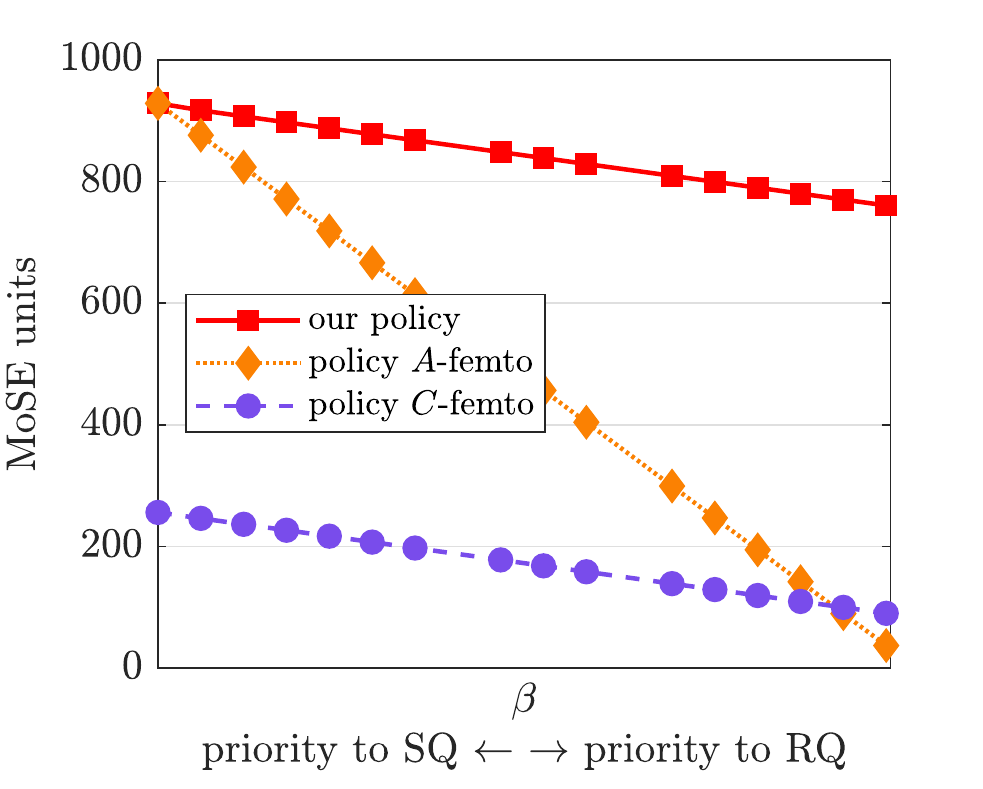}}
\hfil
\subfloat[$\psi$ and $\varphi$ linear]{\includegraphics[width=6cm,trim={0.07cm 0cm 0.55cm 0.1 cm},clip]{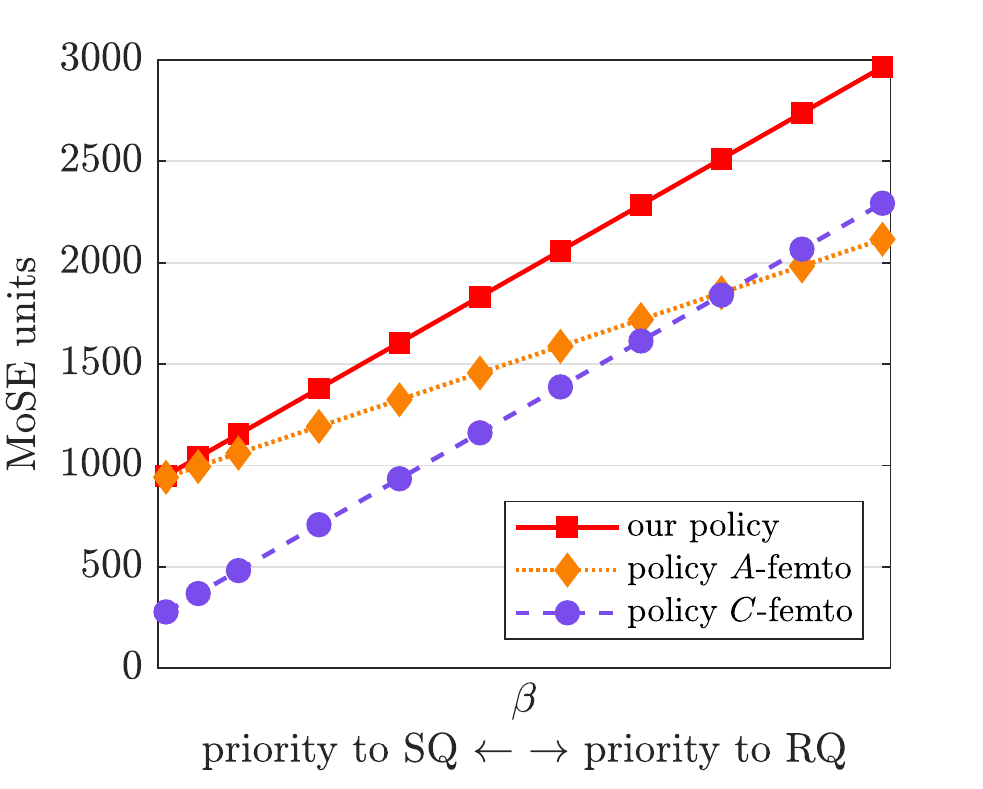}} 
\subfloat[$\psi$ and $\varphi$ logarithmic]{\includegraphics[width=6cm,trim={0.09cm 0cm 0.55cm 0.1 cm},clip]{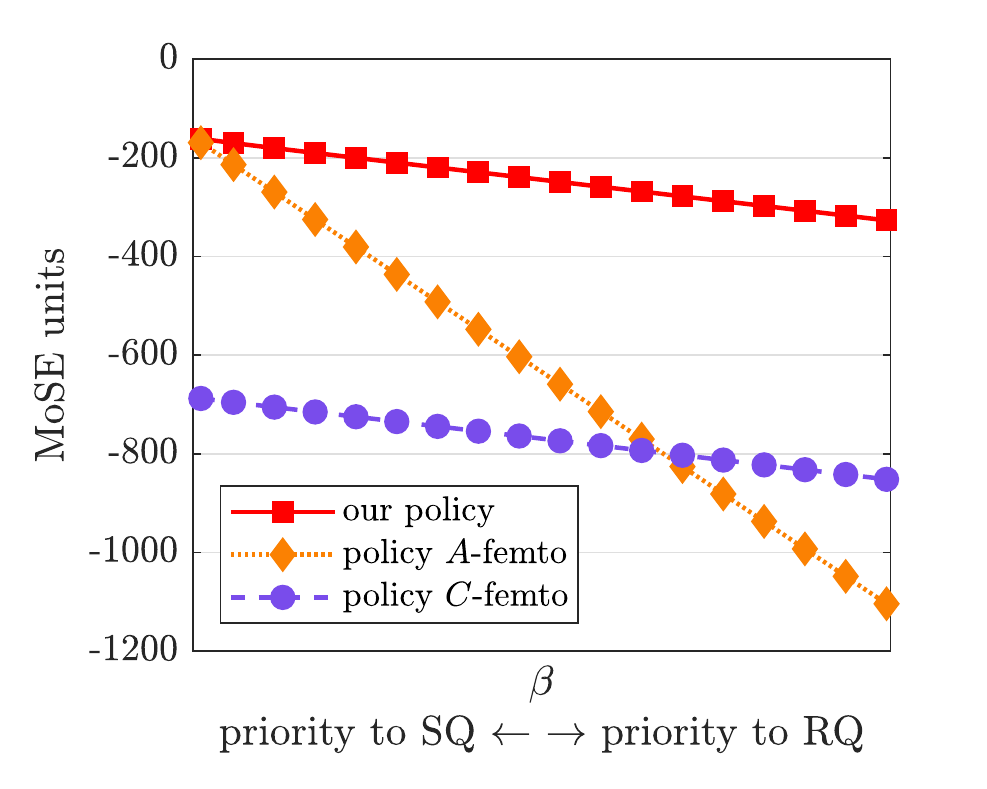}} 
\caption{Scenario 3, MoSE versus $\beta$ for different types of SQ ($s_{uj}=\psi(b_{uj})$) and RQ ($\varphi(r_{ui})$) values/functions. Comparison of our policy with the A-femto and C-femto policies that are based on the algorithm proposed in~\cite{femto_JOURNAL2013}.}  \label{trio2}
\end{figure*}

Figure \ref{trio1}(c) demonstrates that, for values of $\beta$ close to $0$ and values of $r_d$ close to $1$, the performance in SQ  for both policies naturally drops when $r_{min}=0.6$ in comparison to the performance when $r_{min}=0$. This is because fewer contents can be recommended per user and these can largely differ from one user to the next. Therefore, only a few of them can be cached due to the limited cache capacity, and less cache hits will occur. In fact, in the dataset used for this experiment, on average, for every user, only $3\%$ of the catalogue is of utility greater than or equal to $0.6$.

\theoremstyle{definition} \newtheorem{obs_rmin}[obs_optimal]{Observation} 
\begin{obs_rmin} 
Our policy does not choose to radically compromise RQ, leading to similar performance tradeoffs even when additional strict constraints on RQ are imposed.
\end{obs_rmin}

We notice that the tradeoff points of our policy for $r_{min}=0$ and $r_{min}=0.6$ coincide for most of the values of $\beta$, while the maximum observed gain of the latter over the former in terms of RQ is $25\%$. Finally, we observe that even the constraint version of our (close-to-optimal) policy is still able to outperform CAwR with looser constraints on~RQ.

\subsection{Scenario 3}

 So far, we studied scenarios with a single cache in order to be able to compare the performance of the proposed policy with the related work.
We remind the reader that the approximation guarantees of our policy hold for arbitrary networks of caches where users might have access to more than one cache.  The algorithm proposed in \cite{femto_JOURNAL2013} makes caching decisions taking into account such coverage overlaps. However, this problem setup does not contain recommendations. In this scenario, we evaluate the performance, in terms of MoSE, of our policy and some non-joint policies whose caching decisions are made according to \cite{femto_JOURNAL2013}.

We consider a cellular network in a square area of $500$~$m^2$ with $9$ small-cell BS (helpers) and a macro-cell base station (the large cache of our scenario). A total of $100$ users are placed in the area according to a homogeneous Poisson point process (in line with the related works  \cite{femto_JOURNAL2013}, \cite{sermpezis2018sch-jsac}), while helpers are placed in a grid. Helpers' communication ranges are set to $200$~$m$, which results in an average of $3.5$ helpers per user. In this scenario, we will measure the SQ as a function of the  estimated bitrate. More precisely, we assume that $s_{uj}=\psi(b_{uj})$, where $b_{uj}$ are the estimated bitrate that can be supported between user $u$ and cache $j$ and  $\psi$ is an increasing function of $b_{uj}$. Without loss of generality, we assume that the rate from the large cache (or macro-cell cache) $C_0$ is $0.5$ Mbps, while the $b_{uj}$ values for edge caches are chosen  randomly between $2$ and $15$ Mbps\footnote{As we are interested in capturing both wired (CDN) and wireless (femto-caching) scenarios, the physical layer details are beyond the scope of this analysis.}.  In fact, the required Internet connection speed on YouTube \cite{youtube_system_req} is $0.5$ Mbps, and the recommended speed to watch a video in  $4K$ is $20$ Mbps.

We consider a subset of $6000$ unit-sized contents of the Movielens dataset and  $\alpha_u$ and $\varphi$, as in Scenario $2$. We set  the helper's capacity to $1.5\%$ of the catalogue size and $N=5$.  We will compare the performance of our policy, in terms of the MoSE, with two policies that are based on the algorithm proposed in~\cite{femto_JOURNAL2013}:

\theoremstyle{remark} \newtheorem*{lala}{A-femto and C-femto policies}
\begin{lala} 
They generalize the  policies  A and C described in Section \ref{subsec:joint} in a network of caches. They both make the caching decisions based on the femto-caching policy proposed in \cite{femto_JOURNAL2013} that takes into account the fact that users have access to multiple caches in the network. Then, the recommendations part of the policies A and C is applied.
\end{lala}

 For $\psi$ being the identity function, \emph{i.e.,} $s_{uj}=b_{uj}$ and for different values of $\beta>0$,  the achieved MoSE of our policy, the  A-femto, and the C-femto policies are shown in  Fig. \ref{trio2}(a). 
We observe  that, for $\beta$ close to $0$, {\it i.e.}, priority is given to the SQ, the performance of the A-femto policy and our policy coincide. This is because both policies make the same caching and recommendation decisions, {\it i.e.}, cache and recommend the most popular items.
The MoSE achieved by the C-femto policy is lower since, although it provides the best RQ, the recommended items are not necessarily among the cached ones and thus, they need to be retrieved from the large cache at the cost of lower SQ. In fact, this is illustrated in small scale in the toy example in Sec.~\ref{subsec:joint}. As $\beta$ increases,  the priority moves towards RQ, and hence, the performance of the A-femto policy starts to worsen until it is dominated by the one of the  C-femto policy. Our policy continues to perform better than both of them as a result of caching and recommendation orchestration. 
Furthermore, the performance gap between our policy and the C-femto policy remains constant.

\theoremstyle{definition} \newtheorem{obs_multicache}[obs_optimal]{Observation} 
\begin{obs_multicache} 
The performance gains of the proposed policy over  non-joint policies are prominent in generic networks of caches as well.
\end{obs_multicache}

 In the cases studied above, we have considered the SQ and RQ functions ({\it i.e.}, $\psi$ and $\varphi$ respectively) being the identity or the function in \eqref{psi_function_hit} and the logarithmic function respectively. One might wonder how these choices affect the performance of our policy. For this reason, we explore their impact here. We ran on the same dataset as above the  experiment for: {\it i)} both $\psi$ and $\varphi$ being linear functions (Fig.~\ref{trio2}(b)), and {\it ii)} both $\psi$ and $\varphi$ being logarithmic functions (Fig.~\ref{trio2}(c)). We remind the reader that the rationale for the logarithmic function has been elaborated in Sec.~\ref{subsec:QoE}.
We notice that the relative improvement in performance of the proposed policy in comparison to the A-femto and C-femto policies are similar for the different choices of functions considered (Fig.~\ref{trio2}(a)-(c)). We note that a variety of coefficients of these functions have been considered in every case. Furthermore, we observe that even though the range of the y-axis varies in Fig.~\ref{trio2} (a)-(c), the relative gains are similar and the preceding analysis on the relative performance of the three policies holds in every case. 

\theoremstyle{definition} \newtheorem{obs_multicache2}[obs_optimal]{Observation} 
\begin{obs_multicache2} 
The performance improvements are consistent for different choices of SQ and RQ functions.
\end{obs_multicache2}

\section{Related Work}
\label{sec:related}

\textbf{Hierarchical caching.} Optimization of hierarchical caching ({\it e.g.}, CDNs or ICNs) has been widely explored both in the context of wired~\cite{borst2010} and wireless networks~\cite{femto_JOURNAL2013}. Various aspects of this problem have been explored such as caching for different video streaming qualities~\cite{poularakis2014} etc. See, for example, a recent survey on caching in~\cite{JSAC:caching-survey}. Nevertheless, these works are oblivious to the impact of the recommendations, beyond the simple (usually IRM) popularity model used as input. 

\noindent \textbf{Caching-recommendation interplay.} In an early work in this direction~\cite{content-recommendation-swarming}, the authors propose heuristic algorithms for recommendations in P2P networks that take into account both service cost and user preferences. In~\cite{giannakas2018show-me-cache}, the authors propose a recommendation algorithm that tries to bias requests towards cached contents. In a similar spirit,  \cite{cache-centric-video-recommendation} proposes a reordering of the videos appearing in YouTube's related videos section by ``pushing'' on top of the list the cached items. 
However, the caching policy in these works is fixed. 
Considering now different setups, \cite{sermpezis2018sch-jsac} introduces the concept of ``soft cache hits'' that allows the user to choose an alternative cached content if the initially requested is not locally cached. 
Although the caching policies in~\cite{sermpezis2018sch-jsac} are recommendation-aware, the recommender comes \emph{after} the caching decisions. 
 A decomposition algorithm for the joint problem is  proposed in~\cite{chatzieleftheriou2019joint-journal}  for a problem setup closer to our work. 
Targeting cache hit rate maximization, their policy first decides on caching, accounting for the impact of recommendations, and then adjusts the recommendations in order to favor cached items. However, no performance guarantees are given. Similarly, \cite{qi2018optimizing} proposes a decomposition heuristic for the joint caching and recommendation problem. In \cite{liu2019deep}, the authors adopt a different approach and employ machine learning techniques to devise caching and recommendation policies. Finally, the authors in~\cite{Chen:Joint-Globecom18} formulate a joint problem in the somewhat different context of prefetching content over a time-varying channel. 

\noindent \textbf{Joint optimization theory.} Submodularity-based proofs for caching-related problems have flourished since the seminal paper of~\cite{femto_JOURNAL2013}, where the focus is on one set of variables (caching). 
 The decomposition and submodularity method we use is similar in spirit to the methods in~\cite{modiano-optimization-backbone} and \cite{dehghan2016complexity}. While the former  studies quite a different problem than ours, the latter proposes an approximation algorithm for the joint caching and routing problem in cache networks.


\section{Conclusion}
\label{sec:conclusion}

In this paper, we studied the problem of jointly making caching and recommendation decisions in a generic caching network.  This is a problem of great interest as entities like Netflix can now manage both caching and recommendations in their network. To this end, we introduced a metric of user's streaming experience~(MoSE) as a balanced sum of SQ~(affected by the caching allocation) and RQ~(determined by the recommendations the user receives) and we formulated the problem of maximizing users' MoSE. This formulation captures the user's expectations for SQ and RQ from a recommendation-driven application, while, at the same time, allows us to explore the underlying SQ-RQ tradeoffs of the problem. Moreover, the model we considered is generic since SQ can be replaced by any caching gain/profit. We proposed a polynomial-time algorithm that has $\frac{1}{2}$-approximation guarantees (or $\frac{1-1/e}{2}$ in the case of contents of heterogeneous size). Our numerical results in realistic scenarios show important performance gains of our algorithm with respect to baseline schemes and existing heuristics. 
An interesting direction for future work is to consider transmission capacities for the caches and introduce request routing as a variable of the  problem.

\bibliographystyle{IEEEtran}
\bibliography{IEEEabrv,JOURNAL_version-FOR-ARXIV}




\end{document}